\documentclass[a4paper,UKenglish,cleveref, autoref, thm-restate,numberwithinsect]{lipics-v2021}

\pdfoutput=1 
\hideLIPIcs  

\usepackage{algorithmic}
\usepackage[linesnumbered,ruled,vlined]{algorithm2e} 
\usepackage{tabularx}
\usepackage{diagbox}
\usepackage{footmisc}
\usepackage{mathtools}

\usepackage{graphics}
\usepackage{caption}

\usepackage{footmisc}

\newcommand{\NP}{\mbox{\sf NP}}
\newcommand{\union}{\cup}
\newcommand{\prob}[2][]{\text{\bf Pr}_{#1}\left (#2\right)}
\newcommand{\E}{\ensuremath{\mathbb E}}
\renewcommand{\phi}{\varphi}
\newcommand{\eps}{\epsilon}
\newcommand{\poly}{\operatorname{poly}}
\newenvironment{proofof}[1]{\noindent{\bf Proof of #1.}}%
        {\hspace*{\fill}$\Box$\par\vspace{4mm}}

\newenvironment{properties}[2][0]
{\renewcommand{\theenumi}{#2\arabic{enumi}}
\begin{enumerate} \setcounter{enumi}{#1}}{\end{enumerate}\renewcommand{\theenumi}{\arabic{enumi}}}

\newcommand{\jjr}[1]{{#1}}

\title{Bayesian Calibrated Click-Through Auctions\footnote{Authors are listed in $\alpha$-$\beta$ order.}} 

\author{Junjie Chen}{City University of Hong Kong, Hong Kong, China}{junjchen9-c@my.cityu.edu.hk}{}{}

\author{Minming Li}{City University of Hong Kong, Hong Kong, China}{minming.li@cityu.edu.hk}{}{}

\author{Haifeng Xu}{University of Chicago, USA}{haifengxu@uchicago.edu}{}{Supported  by NSF Award  CCF-2303372, Army Research Office Award W911NF-23-1-0030,   Office of Naval Research Award N00014-23-1-2802 and the AI2050 program at Schmidt Sciences (Grant G-24-66104).}

\author{Song Zuo}{Google Research, USA}{songzuo.z@gmail.com}{}{}

\authorrunning{J. Chen, M. Li, H. Xu and S. Zuo} 

\Copyright{Junjie Chen, Minming Li, Haifeng Xu and Song Zuo} 



\category{Track A: Algorithms, Complexity and Games} 

\relatedversion{} 

\nolinenumbers 

\EventEditors{Karl Bringmann, Martin Grohe, Gabriele Puppis, and Ola Svensson}
\EventNoEds{4}
\EventLongTitle{51st International Colloquium on Automata, Languages, and Programming (ICALP 2024)}
\EventShortTitle{ICALP 2024}
\EventAcronym{ICALP}
\EventYear{2024}
\EventDate{July 8--12, 2024}
\EventLocation{Tallinn, Estonia}
\EventLogo{}
\SeriesVolume{297}
\ArticleNo{12}

\begin{document}

\maketitle


\begin{abstract}
    We study information design in \emph{click-through auctions}, in which the bidders/advertisers bid for winning an opportunity to show their ads but only pay for realized clicks. The payment may or may not happen, and its probability is called the \emph{click-through rate} (CTR). 
This auction format is widely used in the industry of online advertising. 
Bidders have private values, whereas the seller has private information about each bidder's CTRs. We are interested in the seller's problem of partially revealing CTR information to maximize revenue. Information design in click-through auctions turns out to be intriguingly different from almost all previous studies in this space since any revealed information about CTRs will never affect bidders' bidding behaviors --- they will always bid their true value per click --- but only affect the auction's \emph{allocation} and \emph{payment} rule. In some sense, this makes information design effectively a constrained mechanism design problem.

 Our first result is an FPTAS to compute an approximately optimal mechanism under a constant number of bidders. The design of this algorithm leverages Bayesian bidder values which help to ``smooth'' the seller's revenue function and lead to better tractability. The design of this FPTAS is complex and primarily algorithmic. Our second main result pursues the design of ``simple'' mechanisms  that are approximately optimal yet more practical. We primarily focus on the two-bidder situation, which is already notoriously challenging as demonstrated in recent works. When bidders' CTR distribution is symmetric, we develop a simple \emph{prior-free} signaling scheme,  whose construction relies on a parameter termed \emph{optimal signal ratio}. The constructed scheme provably obtains a good approximation as long as the maximum and minimum of bidders' value density functions do not differ much. 
\end{abstract}

\section{Introduction}
Many of the phenomenal Internet Technology companies are powered by online advertising \cite{edelman2007internet}. When {an} Internet user browses a webpage, an ad auction {may be} run to determine which ads to be displayed to this user. {Such ad auctions can be either done by the webpage owner or through ad exchange platforms.}
{Large website owners (sellers) sometimes may know the users much better than individual advertisers (bidders) do.} 
It is thus natural for the {seller} to utilize such information advantage to {improve its ad revenue}. 
{As an example, the seller's knowledge about the user can be used to better} predict an Internet user's {\em click-through rate} (CTR), {which is then used in ad auctions to deliver ad impressions more efficiently.}

Interestingly, a recent line of works, starting from Bro Miltersen and Sheffet~\cite{bro2012send}, Emek et al.~\cite{emek2014signaling}, comes to realize that it may not always be the seller's best interest to {use} as much information as possible. That is, {partially use} certain amount -- and the right type --- of information may be more beneficial for the seller's revenue. 
This insight has motivated many studies on the algorithmic problem of optimal \emph{information design} --- i.e., determining what information to be shared with which bidders --- in order to maximize revenue in different auction formats \cite{bro2012send,emek2014signaling,bergemann2022optimal,li2019revenue,fu2012ad,badanidiyuru2018targeting}.

\jjr{While there are two main types of bidding strategies in online advertising: autobidding \cite{deng2022fairness,balseiro2021landscape,balseiro2021robust} and manual bidding, to better understand the optimal design of information in an auction environment, we follow the rich literature and turn to the cleaner manual bidding where the platform incentivizes buyers to tell the truth.} 
However, different from almost all previous works focusing on  stylized auction formats such as the second price auction for independent-value bidders \cite{bro2012send,emek2014signaling,bergemann2022optimal,badanidiyuru2018targeting} and common-value bidders \cite{li2019revenue} and Myerson's optimal auction \cite{fu2012ad}, we focus on a different auction format, i.e., the \emph{click-through auction}, which is the main auction format employed by the current online advertising industry \cite{gspguide}. The click-through auction is also widely known as the generalized second-price auction \cite{edelman2007internet} or position auctions \cite{varian2007position}. We use the term ``click-through auction'' as coined by Bergemann et al.~\cite{bergemann2021calibrated} to emphasize the sale of \emph{clicks}, since this  factor turns out to make information design in click-through auctions significantly different from that in almost all other previously studied auctions.    In a click-through auction for selling a single ad position, each bidder $i$ submits a bid $b_i$ expressing his value for each \emph{click} of his ad. Additionally, the seller {will estimate} a \emph{click through rate} (CTR)   $r_i$ for bidder $i$. The auction runs by ranking bidders according to the \emph{product score} $b_i r_i$, denoted as $b_{(1)} r_{(1)} \geq b_{(2)} r_{(2)} \geq \cdots$, and allocates the item to the bidder with the highest product score $b_{(1)} r_{(1)}$. Crucially, since the auction only sells clicks, the winner does not need to pay \emph{unless a click truly happens} in which case the winner pays the second highest score divided by his own CTR, i.e., $b_{(2)} r_{(2)}/r_{(1)}$. When there is only a single ad slot to allocate, it is straightforward to verify that this auction is truthful, \emph{regardless of the CTR values (even mistakenly estimated)}.\footnote{Such truthfulness holds only when there is a single slot, which is the case we focus on in this paper. Click-through auctions are well-known to be non-truthful and difficult to analyze when there are multiple slots \cite{edelman2007internet}. } As a result, the seller's information advantage of knowing the CTRs \emph{cannot} be exploited to influence the advertisers' bidding behaviors since it is always their best interest to bid the true value per click $v_i$. In some sense, information design here is effectively a naturally restricted format of mechanism design for multiple items with additive bidder values (which is generally a very difficult question \cite{daskalakis2016does}).  This  crucially differs from the information design task in all previously studied auction formats, in which the seller's information about the item (e.g., the CTR) can be revealed to alter bidders' bidding behaviors. This is intrinsically because bidders pay for realized clicks only thus the probability of receiving a click will not matter in click-through auctions. Note that, even in per-impression ad auctions as studied in \cite{badanidiyuru2018targeting},  bidders pay for impressions but care about clicks or conversions,  thus their values do depend on the probability of receiving a click, i.e., the CTR.

Most relevant to ours is the recent work by Bergemann et al.~\cite{bergemann2021calibrated}, who study the same information design question of {optimally} revealing the CTR information in click-through auctions. They focused on a simplified setup with \emph{complete information} about the bidders, i.e., the bidders' values are assumed to be fully known to the seller. A natural \emph{calibration} constraint, originating from Foster and Vohra~\cite{foster1997calibrated}, is imposed on their information design, which simply means the disclosed CTR estimation to each bidder has to be consistent with (i.e., equal in expectation) the true CTR of the bidder. They thus call the new model \emph{calibrated click-through auctions}. Notably, the seller is allowed to privately communicate information with each bidder, which is known as private signaling.  
The main contribution of this paper is to extend the setup of \cite{bergemann2021calibrated} to the more natural and also more widely studied setup of \emph{Bayesian} bidders with independent values. For this reason, We call our model the \emph{Bayesian calibrated click-through auction}.  

To our knowledge, Bergemann et al.~\cite{bergemann2021calibrated} is the first study of information design in auctions in which the   revealed information \emph{cannot} influence bidders' behaviors, but only affect the mechanism itself. This gives rise to an intriguing information design problem --- in some sense, it is even in contrast to one's first impression about information design, also known as \emph{persuasion} \cite{kamenica2019bayesian}, which seeks to exploit  information advantage to influence others' behaviors. In contrast, information design in click-through auctions only affects the {final allocation and payment} of the mechanism but has no effect on bidders' bidding behaviors. From this perspective, information design here is effectively a form of mechanism design. Indeed, a similar phenomenon was observed by Daskalakis et al.~\cite{daskalakis2016does}, who study the \emph{co-design} of information structure and the auction mechanism. They observe that at the optimal co-design, there is no signaling to bidders whatsoever, and the seller will only use the underlying state to decide the item allocation and payment. Our problem is similar, except that we restrict ourselves to the class of click-through auctions. This restriction is motivated by its practical applications.   

Before proceeding to describe our results, we briefly highlight the challenges of information design in our problem.  Indeed,  the difficulty of information design  in strategic games is well documented in previous works \cite{dughmi2014hardness,bhaskar2016hardness}. Even just for auctions, Emek et al.~\cite{emek2014signaling} shows that computing an optimal information design in a second price auction is $\NP$-hard in general for Bayesian bidder values, though it does become polynomial-time solvable in cases with complete information of bidder values. Unfortunately, such tractability does not transfer to the click-through auctions: in a general environment with complete information about bidder values, even the problem of optimal information design in \emph{two}-bidder click-through auctions is left as an open question in \cite{bergemann2021calibrated}. Bergemann et al.~\cite{bergemann2021calibrated} show that when the distributions of CTRs are symmetric, an optimal signaling policy can be characterized. However, back to our Bayesian generalization of their setup, their optimal design for the symmetric information environment is no longer applicable because the design of their signaling scheme relies crucially on knowing the exact identity of the winner for any signal realization, which unfortunately becomes uncertain in our Bayesian setup with random bidder values.
This barrier brings challenges, but also brings opportunities to adopt more algorithmic approaches. Next, we elaborate on our findings. 




\subsection{Results} 
{We study the theoretical aspects of  Bayesian calibrated click-through auctions and develop two encouraging positive results.

Our first result, Theorem~\ref{claimriunderline}, exhibits a Fully Polynomial Time Approximation Scheme (FPTAS) for computing an approximately revenue-optimal signaling scheme in an  \emph{arbitrary} information environment, assuming no point mass in bidders' value distributions. This FPTAS applies to any constant number of bidders. Interestingly, our FPTAS bypasses the notorious challenge of the complete-information general setup of \cite{bergemann2021calibrated}  by relaxing it to the more natural Bayesian valuation setup with smooth value distributions. We use signals of the form $k\epsilon$ for some integer $k$. The key challenge is to enforce the calibration constraint --- i.e., the posterior mean of the CTR conditioned on the signal has to equal the signal itself -- in the rounding process. We develop a nontrivial technique to address the challenge. 
Notably, this FPTAS applies to any auction mechanism, as long as its revenue as a function of signals is Lipschitz continuous. 

Our second main result, Theorem~\ref{distributionfreeappro}, examines the popular recent computational model of unknown value distributions, also known as \emph{prior-free} setups \cite{devanur2013prior}. We give explicit and efficient construction of  simple \emph{prior-free}  signaling schemes for \emph{symmetric} information environments (as studied also by Bergemann et al.~\cite{bergemann2021calibrated}) with strong approximation guarantees.  The constructed signaling scheme is based on a parameter  coined the \emph{optimal signal ratio}, which we find is at most $1$. This allows us to send larger signals (i.e., larger CTR estimations) to the bidder with a high click-through rate than the bidder with a low click-through rate. The proof mainly consists of (i) establishing a connection between the optimal signaling scheme under the unknown value distribution and the optimal signaling scheme under a uniform value distribution and (ii) providing the  approximation ratio of the constructed signaling scheme under a uniform value distribution, which is proved to be $0.995$. En route to proving (ii), we prove a more general result (Proposition \ref{theroemboudningapproximation}) which shows that, if the optimal signal ratio is convex in CTR (which is true for various commonly used distributions), then a signaling scheme with better approximation can be devised.  This result may be of independent interest.}

\subsection{Additional Related Works}
Due to the space limit, we briefly discuss the related works here, while additional discussions in detail are in Appendix \ref{appendix_additional_related_works}. The most relevant literature to our work is information design in auctions.  Information design has been adopted to second price auctions \cite{bro2012send,emek2014signaling,cheng2015mixture,badanidiyuru2018targeting} and Myerson auctions \cite{fu2012ad} in various setups. Recently, Bergemann et al.~\cite{bergemann2021calibrated} provided characterizations for a symmetric calibrated click-through auction but with  complete information of bidder values. Another line of related literature is the sale of information, which  selectively reveals information for revenue improvement. There are a series of works in this line \cite{babaioff2012optimal,chen2020selling,bergemann2018design,cai2020sell,liu2021optimal,zheng2021optimal,li2021selling,bergemann2022selling}. Our work is also related to Bayesian persuasion \cite{kamenica2011bayesian,arieli2019private,babichenko2022multi,dughmi2017algorithmicnoex,babichenko2017algorithmic,xu2020tractability,gradwohl2021algorithms}. We also refer interested readers to a comprehensive survey by Dughmi ~\cite{dughmi2017algorithmic}. \jjr{Finally, our work is partly related to the literature on automated bidding in auctions~\cite{kolumbus2022and,kolumbus2022auctions,feng2023strategic,mehta2023auctions,balseiro2019learning}.}

\section{Preliminaries}
We consider the (arguably more natural) Bayesian version of the \emph{calibrated click-through  auction} of \cite{bergemann2021calibrated}, which is directly motivated by advertising auctions.  At a high level, bidders 
in this auction  are ranked by the products of their \emph{private values} (i.e., the willingness-to-pay per click) and the seller's prediction of \emph{click-through-rate} (\emph{abbr.} CTR).  Bergemann et al.~\cite{bergemann2021calibrated} consider a basic setup in which bidders' private values are perfectly known to the seller. In this paper, we generalize their problem to the Bayesian setup by assuming that bidders' values are independently drawn from distributions. While each bidder knows his own private value, the seller only knows the distributions of bidder values (though our second main result further removes this assumption). This generalization is more aligned with the rich literature of Bayesian auction design, starting from the seminal work of \cite{myerson1981optimal}. Similar to \cite{bergemann2021calibrated}, we also aim to design a revenue-maximizing auction. Different from classic auction design, the seller in a click-through auction has an additional knob to tune, i.e., informing bidders about their CTRs through partial information disclosure (also known as \emph{information design} \cite{kamenica2019bayesian}). So the optimal auction requires designs of both \emph{incentives} and \emph{information}.  

We now describe the general auction setting. There are  $n$ bidders. The private value $v_i$ of bidder (she) $i=1, 2, 3, \dots, n$ is independently drawn from some known distribution $F_i(v_i)$ (with a density function $f_i(v_i)$), over some bounded interval $[a, b]$ with $a, b \in R_{\ge 0} \union \{\infty\}$ and $a < b$. Note that the density functions $f_1(v_1), f_2(v_2), \dots, f_n(v_n)$ are not necessarily identical. As seen, our problem is a natural Bayesian variant of \cite{bergemann2021calibrated}, whose settings can be obtained by allowing $f_i(v_i)$ to be some Dirac $\delta$-function. i.e., a point distribution.


When auctions start, the seller privately observes a CTR vector $r = (r_1, r_2, \dots, r_n) \in [\underline{r}, 1]^n$ for $n$ bidders, which is drawn from a commonly known joint distribution $\lambda(r)$. $\underline{r} >0$ is a small positive value close to $0$. The CTR $r_i$ represents the probability of bidder $i$'s ad being clicked, which may be estimated by the seller with some machine learning method. For bidder $i$, assume the marginal probability $\lambda(r_i) \ge \xi$, where $\xi$ is a small positive value since it does not make sense to consider a CTR with a marginal probability arbitrarily close to $0$ in practice.

To maximize his revenue, the seller may privately disclose some information about $r_i$ to bidder $i$. Suppose the seller has the power of commitment\footref{moredis_ppc_commite_truthful}, and he designs a signaling scheme denoted as $\pi$, which is also observed by the bidders.
Hence, conditioning on observing CTR vector $r$, the seller sends signal $s =(s_1, s_2, \dots, s_n)$ with probability $\pi(s| r)$, where $s_i$ is the signal privately sent to bidder $i$. In the rest of the paper, we assume that $r_i \in \mathcal{R}$ with $\mathcal{R}$ being a discrete and finite set of CTR values. 

\noindent\textbf{Calibrated Click-Through Auction.} The calibrated click-through auction was previously studied by Bergemann et al.~\cite{bergemann2021calibrated}. Initially, each bidder privately observes her value $v_i$ and the seller privately observes the CTR vector $r$. The seller then sends a signal $s_i$ to each bidder $i$ privately. Upon receiving the signal $s_i$, bidder $i$ submits her bid $b_i$ to the auction and the auction determines the winner $i^*$ by selecting the one with the highest product, i.e., $b_{i^*}s_{i^*} = \max_j b_js_j$, and charges the winner for each realized click by 
\[p_{i^*} = \max_{j\neq i^*} \frac{b_js_j}{s_{i^*}}.\]
The winner only pays when a click is received\footnote{\jjr{More discussions about {\it pay-per-click}, commitment power and truthfulness are in Appendix} \ref{append_payperclick} \label{moredis_ppc_commite_truthful}}. Hence, the seller's expected revenue is  $r_{i^*} p_{i^*}$. 
One can easily prove that the auction is truthful\footref{moredis_ppc_commite_truthful} for any $s_1, \ldots, s_n$, as it follows the ``minimum-bid-to-win'' payment rule.
Therefore, without loss of generality, we will assume $b_i \equiv v_i$ in the rest of the paper.

Following  \cite{bergemann2021calibrated}, we adopt the concept \textit{calibration} to  information design. That is, given any private signal  $s_i = s'_i$, bidder $i$'s posterior estimation of CTR $r_i$ should be equal to $s'_i$, i.e.,
\begin{equation}\label{calibratedconstraint}
    E[r_i| s_i = s'_i] = s'_i.
\end{equation}
{The calibration constraint is a consequence of the  revelation principle. Given any signal realization $s_i = s'_i $, bidder $i$ will interpret $s'_i$ through Bayes updates by computing the expected CTR value $   E[r_i| s_i = s'_i]$. The calibration constraint simply requires that the signal itself directly reflects this expected value so that bidders do not need to compute the expected CTR value by themselves. Such revelation-principle-style simplification is widely adopted in information design \cite{kamenica2011bayesian,dughmi2017algorithmicnoex,kamenica2019bayesian}, where signals can without loss of generality be \emph{persuasive}  action recommendations. Analogously, in our auction setup, such direct information scheme will directly signal the posterior mean of the signal, i.e., obeying the calibration constraint \eqref{calibratedconstraint}. 
}
Hence,  any signal $s_i$   must be within $[\underline{r}, 1] \subseteq [0, 1]$ by \eqref{calibratedconstraint}. One important observation in \cite{bergemann2021calibrated} 
is that if the CTRs and signals are discrete, we can rewrite (\ref{calibratedconstraint})  as 
\begin{equation}\label{calibrationconstraintdiscrete}
    \sum_{(r, s): s_i =s_i'} \lambda(r) \pi(s|r) (r_i - s_i') = 0.
\end{equation}

\noindent\textbf{Bayesian Calibrated Click-Through Auction}  We focus on the seller's information design problem, i.e., to design a signaling scheme $\pi$ in order to maximize the seller's expected revenue. The problem is formulated as below,
\begin{equation}\label{ctrsecondproblem}
\begin{split}
  \max_{\pi, s}  & \qquad \sum_r \lambda(r) \sum_s \pi(s | r) \int_v f (v) R(r, v, s) dv \\
  \textrm{subject to}
  & \qquad E[r_i| s_i = s_i'] = s_i' ~~\forall s_i' \in [0, 1], \forall i\\
  & \qquad\pi(s | r) \in [0, 1], \sum_s \pi(s | r) = 1, ~ \forall r, s \in [0, 1]^n,
\end{split}
\end{equation}
where $f(v) = f_1(v_1) f_2(v_2) \cdots f_n(v_n)$ is the density function of $v= (v_1, v_2, \dots, v_n)$, and the revenue $R(r, v, s) = r_{i^*} \frac{\max_{j \neq i^*} v_j s_j}{s_{i^*}}$. 
Since both the signals $s$ and the probabilities of sending signals $\pi(s|r)$ are variables, the objective of  \eqref{ctrsecondproblem} is non-convex. 

{Given CTR vector $r$ and signal $s$, denote the seller's expected revenue as $R(r, s) = \int_v f (v) R(r, v, s) dv$. Specifically,  $R(r, s)$ for $2$ bidders is
\begin{equation}\label{expectedrevenuegivenrsgeneral}
\begin{split}
R(r, s) = &\int_v f (v) R(r, v, s) dv \\
= &\int_a^b r_1  \frac{v_2 s_2}{s_1}\prob{v_1s_1 \ge v_2s_2| v_2} f_2(v_2) dv_2  + \int_a^b r_2  \frac{v_1s_1}{s_2} \prob{v_2s_2 \ge v_1s_1|v_1} f_1(v_1) dv_1,
\end{split}
\end{equation}
where $\prob{v_1s_1 \ge v_2s_2| v_2}$ is the probability of bidder $1$ winning given that bidder $2$'s value is $v_2$ (similarly for $\prob{v_2s_2 \ge v_1s_1|v_1}$). 

By   (\ref{expectedrevenuegivenrsgeneral}),  we can see that our problem differs from \cite{bergemann2021calibrated} in a crucial way}: With any fixed $s$, the winner is also fixed in \cite{bergemann2021calibrated} since they have complete information on $v$, while in our case, either bidder can be the winner with some probability due to the uncertainty in their valuations. This difference somewhat ``smooths" our objective function while at the same time brings new challenges.


\section{Click-Through Auctions in General Environments}\label{sectiongeneralenvironmentdiscrete}

In this section, we present an FPTAS for \eqref{ctrsecondproblem} achieving $1 - O(\epsilon)$ approximation when the number of bidders is a constant. 
This result hinges on a minor continuity assumption on bidders' value distribution. That is, we assume every distribution $F_i$ has no point mass\footnote{The no-point-mass assumption alleviates the tie-breaking problem arising in the case of deterministic bidders' values. A discussion on this assumption is in Appendix \ref{discussiononllipscontinityassumption}.} and has finite second moment, i.e., $\E_{v_i\sim F_i}[v_i^2]<\infty$. Under these two mild technical assumptions, we prove the following theorem. 


\begin{theorem}\label{claimriunderline}
    For any small  $\eps>0$, there is an algorithm that computes a (multiplicative) $1 - O(\eps)$ approximate   signaling scheme in   $\poly(|\mathcal{R}|^n, (\frac{n}{\eps})^n)$ time  where $n$ is the number of bidders and $ |\mathcal{R}|$ is the size of set $\mathcal{R}$ of CTR values.
\end{theorem}

To prove Theorem~\ref{claimriunderline}, we resort to discretizing the strategy space.  Our starting point is a reformulation of Program \eqref{ctrsecondproblem} as a linear program with infinite dimension, and then convert the infinite-dimensional program into a finite-dimensional one by discretizing the signal space. The main technical challenge is to 
prove that the solution to the discrete program approximates the optimal solution to the original infinite-dimensional one. This is more involved than typical rounding approaches and  hinges on the following two key properties: 
\begin{enumerate}
  \item[a).] For any signaling scheme, the value of the discrete objective is sufficiently close to the original objective. This is a consequence of the following Lipschitz continuity of revenue function $R(r, s)$, whose proof is deferred to Appendix \ref{appendixproofoftheorem31claim}. 

\begin{lemma}\label{continousfunctionRrs}
$R(r,s)$ is Lipschitz continuous in $s \in [\underline{r}, 1]^n$ with some constant $C_n$, for any given CTR $r = (r_1, r_2, \dots, r_n)$. 
\end{lemma}

  \item[b).] For any solution to the original infinite dimensional program, there exists a corresponding feasible solution to the discrete program that is sufficiently close to the original solution. 
\end{enumerate}
These two properties together can ensure the existence of a solution to the discrete program yielding the revenue that approximates the optimal solution of the original infinite dimensional program. Thus by solving the discrete program, we can get an approximately optimal solution.

{

{\em The major challenge} here is to prove property b), as any na\"{i}ve rounding of a signaling scheme from the continuous signal space to a discrete space would break the calibration constraints and hence end up with an infeasible rounded signaling scheme to the discrete program. To circumvent the difficulty, we introduce a novel technique to retain all the calibration constraints.

At a high level, we will reserve a small amount of probability mass from the given signaling scheme at the beginning and only apply rounding to the remaining probability mass. After the rounding step, we will redistribute the reserved probability mass to carefully fix all the calibration constraints broken by the rounding step. In particular, one has to be extremely careful to avoid the straightforward discretization structure of the signal space, because otherwise a large fraction of probability mass reservation will be needed for the fixing stage, which leads to a significant revenue loss and fails the $1 - O(\epsilon)$ approximation. The detailed proof of Theorem \ref{claimriunderline} is involved and relegated to Appendix \ref{appendix_proofof_theorem_section3main}.
}

\medskip
It is worthwhile  to compare our FPTAS in Theorem \ref{claimriunderline}  with a recent FPTAS  for  multi-channel Bayesian persuasion by Babichenko et al.~\cite{babichenko2022multi}, and   highlight their key differences.  While both rely on the continuity of the sender's utility functions, their     techniques are significantly different. Specifically,  the FPTAS of \cite{babichenko2022multi} applies  to a constant number of states but many receivers (i.e., bidders) whereas our FPTAS applies to many states but a constant number of bidders. Their FPTAS discretizes the space of  distributions over the states (thus requires a constant number of states) whereas our FPTAS discretizes the space of posterior CTR mean, which is why we have to adjust the scheme to satisfy the calibration constraints. Our choice of discretizing posterior means is due to its direct usage in click-through auctions. Such a special structure and the resultant challenges of fixing calibration constraint violation --- which is the core difficulty in our proof ---  is not present in the setup of \cite{babichenko2022multi}.



We remark  that the discretization technique applied to Theorem \ref{claimriunderline} is in fact regardless of the expected revenue $R(r, s)$. To retain the calibration constraint and preserve the approximation simultaneously, one basic requirement is that $R(r, s)$ needs to be Lipschitz-continuous in signal $s$. Therefore, our algorithm can be treated as a framework that can accommodate different auctions (with Lipschitz continuous expected revenue $R(r, s)$) and maintain the calibration of signals. Based on this observation, we consider one simple modified auction: a click-through auction with a reserved price. Reserve prices are widely studied in the literature and commonly used in the industry \cite{paes2016field,ostrovsky2011reserve}. {Similarly, we only need to show $R(r, s)$ of this modified auction is Lipschitz-continuous for Corollary \ref{corollaryclickthroughreservedprice} to be true, which is deferred to Appendix \ref{Proof_of_Corollary_corollaryclickthroughreservedprice}.}


\begin{corollary}\label{corollaryclickthroughreservedprice}
    In a click-through auction where winner pays at least some fixed reserve price $p$, for any constant number of bidders and any small $\eps>0$, there is an algorithm that computes  (multiplicatively) $1 - O(\eps)$ approximate signaling scheme with  time complexity $poly(|\mathcal{R}|, \frac{1}{\eps})$.
\end{corollary}

\section{Click-Through Auctions in Symmetric Environments}\label{section4symmetricenvironments}
While Theorem \ref{claimriunderline} provides an algorithm  that can compute an approximately optimal signaling scheme in fairly general setups, the algorithm relies on solving large-scale linear programs  thus may be too costly in reality and also lacks interpretability. In this section, we pursue the design of  ``simple'' signaling schemes through the approximation lens. Given the challenge of the problem in general, we shall restrict our attention to a fundamental special case in this section ---  i.e.,   $2$ bidders and symmetric environments ---which is also a major focus of our {preceding} work by Bergemann et al.~\cite{bergemann2021calibrated}. 
\jjr{
For this basic case, we first characterize that without calibration constraint, the optimal information design is governed by a single parameter called \emph{optimal signal ratio}, which is at most $1$. Such an observation motivates an explicit construction of a simple and \emph{prior-free} signaling scheme, which is close to the optimal when bidder values' probability density function does not fluctuate much.}

An environment is \textit{symmetric} \cite{bergemann2021calibrated} if 1) the distribution of the products $v_1r_1, v_2r_2, \dots, v_nr_n$ are exchangeable, and 2) the values $v_i$'s are i.i.d. drawn from a distribution $F(v)$. Hence, the symmetric environment further implies that the distribution of $r_1, r_2, \dots, r_n$ are also exchangeable, i.e., $\lambda(r_1, r_2, \dots, r_n) = \lambda(perm(r_1, r_2, \dots, r_n))$ where $perm(\cdot)$ is any permutation function. For the purpose of easy exposition, we allow $r_i$ to be drawn from a slightly enlarged interval $[0, 1]$, with which the results obtained hold for $r_i\in[\underline{r}, 1]$ as in our original setting.


\begin{definition}
The distribution $F(v)$ (with density $f(v)$) has a monotone hazard rate (MHR) if the ``hazard rate'' of the distribution $\frac{f(v)}{1-F(v)}$ is increasing in $v$. 
\end{definition}

 MHR is a standard assumption widely used in economics  \cite{badanidiyuru2018targeting,hartline2009simple,chawla2007algorithmic}. Many distributions are MHR, such as uniform distribution, exponential distribution, gamma distribution, etc. Note that MHR implies regularity  proposed by  Myerson~\cite{myerson1981optimal}.


\begin{theorem}\label{distributionfreeappro}  
There exists a prior-free signaling scheme {that can be found in polynomial time and} guarantees a (multiplicative) $0.995 \cdot (\underline{f}/ \overline{f})^2$-approximation for any continuous MHR distribution $f(v)$,
where $\overline{f}$ and $\underline{f}$ are the respective maximum and minimum values of $f(v)$ for $v\in [0, c]$ with any $c>0$.
\end{theorem}



{{The approximation ratio in Theorem \ref{distributionfreeappro} achieves its best  ratio $0.995$ for uniform distribution.} {More generally, the approximation ratio depends on the term $\underline{f}/\overline{f}$ which intuitively captures 
how much the distribution deviates from being uniform. Notably, both the algorithm and the signaling scheme  require no knowledge about the value distribution $F$,  i.e., $F$ could be unknown to the seller (though Section \ref{knowndistributionapproximation} provides a better guarantee when $F$ is known and satisfy certain  properties).} 
We remark that the signaling scheme in Theorem \ref{distributionfreeappro} is \emph{simple} in the sense that its construction is parameterized by a single parameter called \emph{optimal signal ratio}. 
\jjr{The information design problem thus can be easily optimized by empirically selecting a value that gives the best performance.} 
More importantly, Theorem \ref{distributionfreeappro} also provides a robust solution to the problem, which can be very useful when it is hard or costly to accurately estimate the distributions of bidder values.
}

\subsection{{Key Primitive of the Construction: the Optimal Signal Ratio}}

Before presenting the proof, we discuss the main ingredient for the construction of the signaling scheme, the \emph{optimal signal ratio}. Given a CTR vector $r = (h, l)$ (without loss of generality, assume $h \ge l$) and a signal pair $s = (s_1, s_2)$, the expected revenue by (\ref{expectedrevenuegivenrsgeneral}) can be rewritten as ({Recall $f(v) = 0, \forall v \notin [0, c]$. Note that although the seller requires no knowledge about $F(v)$ to construct the scheme, the bidders' value should follow some prior $F(v)$.})
\begin{small}
\begin{equation}\label{expectedrevenuegivenrs}
\begin{split}
R(r, s)
=  &\int_0^c (h  \frac{v_2 s_2}{s_1})\int_{\frac{v_2 s_2}{s_1}}^c f(v_1) dv_1 f(v_2) dv_2 + \int_0^c (l  \frac{v_1s_1}{s_2} )\int_{\frac{v_1 s_1}{s_2}}^c f(v_2) dv_2 f(v_1) dv_1.
\end{split}
\end{equation}
\end{small}Clearly, from (\ref{expectedrevenuegivenrs}), the expected revenue is determined by the signals through their ratio $x = \frac{s_2}{s_1}$, which we call the \emph{signal ratio}. 
We define the \emph{optimal signal ratio} to be the signal ratio maximizing (\ref{expectedrevenuegivenrs}).
In the rest of the paper, without loss of generality, we assume $h=1$. Then the optimal signal ratio only depends on the smaller CTR $l$, denoted as $x(l)$ with $l \in [0, 1]$. 
\begin{lemma}\label{lemmatemperary}
    The optimal signal ratio $l < x(l) \le 1$ for $l \in [0, 1)$ and $x(1) = 1$.
\end{lemma}


\jjr{
Lemma \ref{lemmatemperary} is  directly implied by Lemma \ref{lemmas2s1lessthan1} and  \ref{propertyxlconstruction} in Appendix \ref{app_propoertyofoptimalsignal}, where more discussion about $x(l)$ can be found.  There are two interesting implications from Lemma \ref{lemmatemperary}. The first interesting implication is that without the calibration constraint, it would be optimal for the seller to only send  pairs of signals with the optimal signal ratio where a larger (\emph{resp.} smaller) signal is observed by the bidder with a higher (\emph{resp.} lower) CTR, which motivates our construction of the signaling scheme. The second implication is that by $l<x(l)$ in Lemma \ref{lemmatemperary}, information design induces more intensive competition between buyers by partial-information revelation (with signal ratio $x(l)$) than  full revelation (with signal ratio $l$).  This again explains the  observations in \cite{badanidiyuru2018targeting} that revenue extraction in auctions needs competition and too ``fine-grained'' targeting information may lead to a thin market. More detailed discussions on the second implication are in Appendix \ref{appen_discuss_implic_osigratio}.}

{As an   application of   Lemma \ref{lemmatemperary}, we observed that $x(l) = 1$ for $l \in [0, 1]$ for any \emph{exponential distribution}. It  turns out that in this special case, the revenue-maximizing calibrated signaling scheme   is simply to reveal no information, i.e., always sending one calibrated signal pair with $s_1=s_2=E[r_i]$. This leads to the following proposition, whose proof is deferred to  Appendix \ref{proof_of_propositionexponentialdistribution}.   }

\begin{proposition}\label{propositionexponentialdistribution}  
If the value distribution is an exponential distribution with density  $f(v) = \lambda e^{-\lambda v}, v\ge0, \lambda>0$, the  signaling scheme revealing no information is   revenue-optimal.
\end{proposition}

\begin{algorithm}[t]
\SetAlgoNoLine
\begin{enumerate}
\item  Given CTR vector $r=(1, l)$, we compute a list of candidate signals $\sigma_0, \sigma_1, \dots, \sigma_{K-1}$, where $K = \lfloor \log_{x(l)} l \rfloor$ such that $x(l)^{K+1} \le l\le x(l)^K$. The signal $\sigma_i$ is computed as
\begin{equation}
    \sigma_K = 1, \sigma_0 = x(l)^K, {\sigma_i} = \sigma_{i-1}\cdot x(l), \forall i \in [K].
\end{equation}
\item  Define $p(r, s) = \lambda(r)\pi(s|r)$ as the probability mass of sending signal $s$ conditioning on CTR $r$. Let $(l, 1)$ and $(1, l)$ send signals $(\sigma_k, \sigma_{k+1})$ and $(\sigma_{k+1}, \sigma_{k})$, respectively, with the same probability mass 
\[p\Big((l, 1), (\sigma_k, \sigma_{k+1})\Big) = p\Big((1, l), (\sigma_{k+1}, \sigma_{k})\Big) = p_k, \forall k \{ 1, 2, \dots, K-1\}\]
where $p_k$ is computed as 
\begin{equation}\label{xkprobabilitymass}
p_k = p_{k-1}\cdot \frac{1-\sigma_k}{\sigma_k - l} = p_0 \prod_{i=1}^k \frac{1-\sigma_i}{\sigma_i -l}, \forall k \in \{1, 2, \dots, K-1\}.
\end{equation}
\item  By the above construction, we know  $p\Big((l, 1), (\sigma_{K-1}, 1)\Big) = p\Big((1, l), (1, \sigma_{K-1})\Big) = p_{K-1}$.  To maintain calibration constraint, the seller will additionally send signal $(\sigma_0, \sigma_0)$ with $p\Big((l, 1), (\sigma_0, \sigma_0)\Big) = p\Big((1, l), (\sigma_0, \sigma_0)\Big) = z$, where 
\begin{equation}\label{relationzx0}
z = p_0\cdot \frac{\sigma_0 - l}{l+1 - 2\sigma_0}.
\end{equation}
\item Choose a proper $p_0$ (other parameters $z$, $p_i$ are then determined) so that 
\[z+p_0+p_1+\cdots+p_{K-1}=\lambda(r)=\frac{1}{2}\]
\end{enumerate}
\caption{Construction of \emph{Simple} Signaling Scheme}
\label{alg_construction_sgscheme}
\end{algorithm}

\subsection{Construction of the \textit{Simple} Signaling Scheme}\label{sectionsconstructingsignalingscheme}

{
In this part, we present the construction of a simple signaling scheme. Recall that given the smaller CTR $l$, the maximum expected revenue is achieved when $\frac{s_2}{s_1} = x(l)$, i.e., the optimal signal ratio. In other words, one can construct an approximately optimal signaling scheme by sending signal pairs of the optimal signal ratio $x(l)$ as frequently as possible, while maintaining the calibration constraints.

Based on the above intuition, we then present the construction of the simple signaling scheme $\pi$ depicted in Algorithm 
\ref{alg_construction_sgscheme}. Since the CTRs of bidders are exchangeable, we can separately consider the signaling scheme for each pair of CTR vectors. Without loss of generality, consider a pair of CTR vectors $(l, 1)$ and $(1, l)$ such that $\lambda(r = (l, 1)) = \lambda(r = (1, l)) = \frac{1}{2}$. In particular, we assume $x(l) < 1$ (otherwise, optimality can be easily achieved by revealing no information).
With properly chosen parameter $p_0$, the constructed signaling scheme is calibrated and valid (see Appendix \ref{app_verifyingsignalingscheme}). Note that by (\ref{relationzx0}), the probability mass $z=0$ if $\sigma_0 = l$, which implies all the signals sent are of the optimal signal ratio. Therefore,  the optimal expected revenue is achieved. 

}


{

\begin{remark}
As mentioned previously, the construction only depends on one parameter, the optimal signal ratio $x(l)$. Besides, the scheme has a clear structure, i.e., computing a geometric series of signals $\{\sigma_i\}$ and then assigning probabilities, both of which can be done efficiently.
\end{remark}

{We present the following  concrete example that describes the constructed signaling scheme.}

\begin{table*}[t]
    \centering
        \caption{Example of a signaling scheme. Each entry corresponds to $\lambda(r)\pi(s | r)$, i.e., the probability mass of sending signals $s = (\sigma_i, \sigma_j)$ when observing the CTR vector $r$.}
\begin{tabular}{|c|c|c|c|c|c|c|c|c|c|}
\hline \diagbox{r}{s}&($\sigma_0$, $\sigma_0$)&($\sigma_0$, $\sigma_1$)&($\sigma_1$, $\sigma_2$)&($\sigma_2$, $\sigma_3$)&($\sigma_3$, $\sigma_4$)&($\sigma_1$, $\sigma_0$)&($\sigma_2$, $\sigma_1$)&($\sigma_3$, $\sigma_2$)&($\sigma_4$, $\sigma_3$)\\
\hline 
(0.6, 1)&$z$&$p_0$&$p_1$&$p_2$&$p_3$&0&0&0&0\\
\hline (1, 0.6)&$z$&0&0&0&0&$p_0$&$p_1$&$p_2$&$p_3$\\
\hline
\end{tabular}

    \label{tab:my_label}
\end{table*}

\begin{example}
Let $f(v)$ be the uniform distribution over $[0, 1]$. The click-through rates are $h = 1$ and $l = 0.6$, with 
$\lambda((h, l)) = \lambda((l, h)) = 0.5$. By solving (\ref{expectedrevenuegivenrs}), we obtain the optimal signal ratio $x = \frac{9}{10}$. We consider $k=4$ levels of signals within $[0.6, 1)$ as $\sigma_0 = (\frac{9}{10})^4, \sigma_1 = (\frac{9}{10})^3, \sigma_2 = (\frac{9}{10})^2, \sigma_3 = \frac{9}{10}, \sigma_4 = (\frac{9}{10})^0 = 1$. The signaling scheme is given by Table \ref{tab:my_label}, where the probabilities $z$ and $p_0, p_1, p_2, p_3$ are determined according to \eqref{xkprobabilitymass} and \eqref{relationzx0} to keep the construction a calibrated signaling scheme.

Note that all signal pairs except $(\sigma_0, \sigma_0)$ follows the optimal signal ratio $x = \frac9{10}$. Hence the revenue suboptimality only happens when sending the signal $s = (\sigma_0, \sigma_0)$, of which the probability mass is $z \approx 0.016$. 
It turns out that the expected revenue of this calibrated signaling scheme is $0.2698$. In contrast, by relaxing the calibration constraint, the maximum revenue one can achieve is $0.27$ (i.e., always sending signals of the optimal signal ratio $9/10$). Since $0.27$ is clearly an upper bound of the optimal revenue for any calibrated signaling scheme, the revenue loss of our construction is less than $0.075\%$.


We remark that the revenue loss decreases very quickly as the number of signals $k$ increases. In this example, if $l = 0.2$ instead, then the optimal signal ratio $x(l) = \frac45$ and we can choose $k=7$, yielding a much smaller probability mass of sending suboptimal signals $z \approx 1.5 \times 10^{-5}$.

\end{example}
}

\subsection{Proof of Theorem \ref{distributionfreeappro}}\label{distributionfreepoorf}

{We first show a special case of Theorem \ref{distributionfreeappro}, which is also a cornerstone for the main proof. That is, 
if the value distribution is   uniform   (i.e., $ \underline{f}/ \overline{f} = 1$), we can design a $0.995$-approximate scheme.
\begin{proposition}\label{theoremunformdistribution}
Given a uniform distribution  supported on $[0, c]$, the constructed signaling scheme can achieve at least (multiplicative) $0.995$-approximation.
\end{proposition}
Simple calculation leads to the optimal signal ratio $x(l)=\frac{3+l}{4}$ under the uniform distribution. The proof of Proposition \ref{theoremunformdistribution} {is a combination of} the following two lemmas. 
Recall that the constructed signaling scheme sends signal pairs of {the optimal} signal ratio $x(l)$ {with probability $1 - z$} and {with the remaining probability $z$ sends the suboptimal signal pair $(\sigma_0, \sigma_0)$}. Lemma \ref{proposition46_lemma_1} {bounds the probability $z$}, while Lemma \ref{proposition46_lemma_2}
{lower bounds the revenue approximation when sending out $(\sigma_0, \sigma_0)$}. We will discuss the proof of Lemma \ref{proposition46_lemma_1} {in Section \ref{knowndistributionapproximation}}, where a more general result is proved, and defer the proof of Lemma~\ref{proposition46_lemma_2} to Appendix \ref{proof_of_theoremunformdistribution_lemma}.
\begin{lemma}\label{proposition46_lemma_1}
For optimal signal ratio function $x(l)=\frac{3+l}{4}$, there exists
a signaling scheme whose worst-case approximation is $1-z^*$, where $z^* \leq 0.04$.
\end{lemma}
\begin{lemma}\label{proposition46_lemma_2}
  Sending signal $(\sigma_0, \sigma_0)$, i.e., the signal ratio is $1$,  is (multiplicatively) $\frac{8}{9}$-approximation.
\end{lemma}

\begin{proofof}{Proposition~\ref{theoremunformdistribution}}
{Combining Lemma \ref{proposition46_lemma_1} and Lemma~\ref{proposition46_lemma_2}, the approximation of the constructed signaling scheme is at least} $(1-z^*) \cdot 1 + z^*\cdot \frac{8}{9} \geq 224/225 \approx 0.995$.
\end{proofof}

}

{We prove Theorem~\ref{distributionfreeappro} by showing that the signaling scheme we constructed for the uniform distribution  is a $0.995 \cdot (\underline{f}/ \overline{f})^2$-approximation for any value distribution. In particular, sending signal pairs of signal ratio $\frac{3+l}{4}$ is $(\underline{f}/ \overline{f})^2$-approximate and sending $(\sigma_0, \sigma_0)$ is $\frac{8}{9} \cdot ( \underline{f}/ \overline{f})^2$-approximate. An analysis similar to Proposition \ref{theoremunformdistribution} finally leads to $0.995 \cdot (\underline{f}/ \overline{f})^2$ approximation.}

\begin{proofof}{Theorem \ref{distributionfreeappro}}
We first derive the upper and lower bound of the expected revenue $R(r, s)$ by connecting to the uniform-distribution case. Given a CTR $r = (1, l)$ and any signal pair $s = (s_1, s_2)$ with ratio $\frac{s_2}{s_1} \le 1$ (by Lemma \ref{lemmatemperary},  $\frac{s_2}{s_1} \le 1$ is without loss of generality), the upper bound by (\ref{expectedrevenuegivenrs}) is computed as
\begin{small}
\begin{equation*}
\begin{split}
R(r, s) 
=& \int_0^c \frac{v_2 s_2}{s_1}\int_{\frac{v_2s_2}{s_1}}^c f(v_1)dv_1 f(v_2) dv_2 + \int_0^{c\cdot\frac{s_2}{s_1}} (l \cdot \frac{v_1s_1}{s_2} )\int_{\frac{v_1s_1}{s_2}}^c f(v_2)dv_2 f(v_1) dv_1\\
\le& \int_0^c  \frac{v_2 s_2}{s_1}\int_{\frac{v_2s_2}{s_1}}^c \overline{f} dv_1 \overline{f} dv_2 + \int_0^{c\cdot\frac{s_2}{s_1}} (l\cdot \frac{v_1s_1}{s_2} )\int_{\frac{v_1s_1}{s_2}}^c \overline{f} dv_2 \overline{f} dv_1\\
=&(\overline{f}c)^2\cdot \Big( \int_0^c  \frac{v_2 s_2}{s_1}\int_{\frac{v_2s_2}{s_1}}^c \frac{1}{c} dv_1 \frac{1}{c}  dv_2 + \int_0^{c\frac{s_2}{s_1}} (l \cdot \frac{v_1s_1}{s_2} )\int_{\frac{v_1s_1}{s_2}}^c \frac{1}{c} dv_2 \frac{1}{c} dv_1 \Big)\\
=& (\overline{f}c)^2\cdot\overline{R}(r, s)
\end{split}
\end{equation*} 
\end{small}where $\overline{R}(r, s)$ computes the expected revenue under a uniform value distribution given CTR vector $r$ and signal pair $s$. Similarly, the lower bound is $R(r, s) \ge (\underline{f}c)^2\cdot\overline{R}(r, s)$.
Since $R(r, s)$ is only related to the signal ratio $x=\frac{s_2}{s_1}$, we rewrite $R(r, s)$ as $R(r, x)$. {Similarly,  rewrite $\overline{R}(r, s)$ as  $\overline{R}(r, x)$.} 
Let $x^* \le 1$ and $\overline{x} \le 1$ (by Lemma \ref{lemmatemperary}) be the optimal signal ratio for $R(r, x)$ and $\overline{R}(r, x)$, respectively. Then, 
\begin{gather*}
(\underline{f}c)^2\cdot\overline{R}(r, x^*) \le R(r, x^*) \le (\overline{f}c)^2\cdot\overline{R}(r, x^*) \\(\underline{f}c)^2\cdot\overline{R}(r, \overline{x}) \le R(r, \overline{x}) \le (\overline{f}c)^2\cdot\overline{R}(r, \overline{x})
\end{gather*} 
Simple calculation leads to $ (\underline{f} / \overline{f}  )^2\cdot R(r, x^*) \le (\underline{f} c)^2\cdot \overline{R}(r, x^*) \le (\underline{f}c )^2\cdot \overline{R}(r, \overline{x})  \le R(r, \overline{x})$. The second inequality is due to the optimality of $\overline{x}$ to $\overline{R}(r, x)$. 
The whole inequality implies that if the seller sends signal pair whose ratio is optimal under a uniform value distribution, the approximation ratio is at least $ (\underline{f} / \overline{f} )^2$.
Notice that Lemma  \ref{proposition46_lemma_2} shows that  under a uniform value distribution, sending a signal pair with signal ratio $\frac{s_2}{s_1}=1$ achieves $\frac{8}{9}$ approximation, i.e., $\frac{8}{9}\overline{R}(r, \overline{x}) \le \overline{R}(r, 1)$. Hence, we have $\frac{8}{9}(\underline{f} / \overline{f}  )^2\cdot R(r, x^*) \le \frac{8}{9}(\underline{f}c )^2\cdot \overline{R}(r, \overline{x})  \le (\underline{f}c )^2\cdot \overline{R}(r, 1) \le R(r, 1)$, implying that if the seller sends signal pair with ratio equal to $1$, then the approximation is $\frac{8}{9}\cdot(\underline{f} / \overline{f}  )^2$. 

Now, we present the $0.995 \cdot ( \underline{f}/ \overline{f})^2$ approximation. Lemma \ref{proposition46_lemma_1} shows that under a uniform distribution, the probability mass $z$ of sending signal $(\sigma_0, \sigma_0)$ is upper bounded as $z < 0.04$, while $1-z > 0.96$ probability mass is for sending signal pairs of the optimal signal ratio $\overline{x}$.   Hence, if we construct a signaling scheme by assuming the unknown value distribution to be a uniform distribution, it can achieve at least $0.995 \cdot ( \underline{f}/ \overline{f})^2$-approximation, where $0.995\approx 0.04\times \frac{8}{9} + 0.96$.
\end{proofof}

The above proof implies that the \emph{prior-free} signaling scheme is constructed with signal ratio $x(l)=\frac{3+l}{4}$ obtained by assuming the unknown value distribution to be a uniform distribution. 
One direct result from the proof is that an easier but loose scheme is revealing no information which gives $\frac{8}{9}\cdot(\underline{f} / \overline{f}  )^2$ approximation. The improved ratio in Theorem \ref{distributionfreeappro} demonstrates the benefits of strategic information revelation.

\subsection{{Completing the Last Piece ---  Approximation Guarantee with Known Distributions}}\label{knowndistributionapproximation}
{The only missing piece for completing the proof of Theorem \ref{distributionfreeappro} is the proof of Lemma \ref{proposition46_lemma_1}, which is for a special linear optimal signal ratio function $x(l)=\frac{3+l}{4}$. In this part, we {prove a more general result for any convex $x(l)$, which directly implies Lemma~\ref{proposition46_lemma_1} with linear $x(l)$.}
}
{To formally state the general proposition for convex $x(l)$, we need the following definitions.}


\begin{definition}\label{definitioninitiaointersection}
 Given that the optimal signal ratio $x(l)$ is a convex function in $l \in [0, 1]$, define the following notations.
\begin{enumerate}
    \item Initial number $K_0$:  it satisfies that i) $x=x(l)^{K_0} \ge l$, and ii) $x=x(l)^{K_0+1}$ crosses the line $x=l$ and intersects at some point $(l, l)$ with $l< 1$.
    \item Define $S(k, l)$  with $\sigma_0 = x(l)^k$, $\sigma_i = x(l)^{k-i}$ as 
    \begin{equation}\label{equantitysk}
    S(k, l) = 1 + \frac{l+1-2\sigma_0}{\sigma_0 -l} + \frac{l+1-2\sigma_0}{\sigma_0 -l}\cdot \frac{1-\sigma_1}{\sigma_1 -l} +\cdots + \frac{l+1-2\sigma_0}{\sigma_0 -l} \cdot \prod_{i=1}^{k-1}\frac{1-\sigma_i}{\sigma_i -l}
\end{equation}
    \item Intersection point $l[k]$: given some $k > K_0$, $l[k] \neq 1$ is a solution to the equation $x(l)^{k} = l$. In another word, $x=x(l)^{k}$ crosses the line $x=l$ and intersects at the point $(l[k], l[k])$.
\end{enumerate}
\end{definition}

Recall in the construction of the signaling scheme that given a CTR vector $r = (l, 1)$, at most $K = \lfloor \log_{x(l)} l \rfloor$ signals (i.e., $\sigma_i$) are within the interval $[l, 1]$. In fact, the initial number $K_0$ specifies the minimum number of signals constructed within $[l, 1]$ for any $l < 1$. Also, the probability mass $z$ of sending $(\sigma_0, \sigma_0)$ can be computed with the defined  $S(k, l)$,
\begin{equation}\label{computationzwithssss}
     z + p_0 + p_1 + \cdots + p_{K-1} = \frac{1}{2} 
    \quad \Longleftrightarrow \quad z\cdot S(K, l) = \frac{1}{2}
\end{equation}


With the above definition, we  present the general result. 
\begin{proposition}\label{theroemboudningapproximation} 
Assume $x(l)$ is a convex function.  There exists
a signaling scheme whose worst-case approximation is $1-z^*$, where $z^* = 1/S(K_0,  l[K_0 + 1])$ and $S(K_0,  l[K_0 + 1])$ is computed as (\ref{equantitysk}). 
\end{proposition} 

Proposition \ref{theroemboudningapproximation} provides a \emph{worst-case} approximation bound for the constructed signaling scheme and shows how it depends on the intrinsic properties of the optimal signal ratio function. Intuitively, by Equation (\ref{equantitysk}), two situations lead to a small $z^*$: 
1) The initial number $K_0$ is large. In this case, it implies that the seller can send many signal pairs with the optimal signal ratio and thus  $S(K_0, l[K_0+1])$ grows exponentially; 
2) $l[K_0+1]$ is close to $\sigma_0$ in (\ref{equantitysk}).  {Later, we will show it is equivalent to that $l[K_0+1]$ is close to $x(l[K_0+1])^{K_0}$. This case will lead to a large $S(K_0, l[K_0+1])$  and thus small $z^*$.} 
In fact, given the number of signals $k\ge K_0$ and the CTR $l$, the approximation of the constructed signaling scheme can be calculated similarly as $1-z$ with $z = 1/S(k, l)$. As we will prove later, the worst-case approximation is obtained only when $l$ is approaching $l[K_0 + 1]$ from above (see Figure \ref{figobservationaaaproofB} for an example). Hence, when $l$ varies, the actual approximation of the constructed scheme may be (much) better than $1-z^*$.


Proposition \ref{theroemboudningapproximation} may be of independent interest. If an estimate for the bidders' value distribution is available, the seller then can construct a signaling scheme as in Section \ref{sectionsconstructingsignalingscheme} that has an approximation guarantee indicated by Proposition \ref{theroemboudningapproximation}. Note that Proposition \ref{theroemboudningapproximation} only requires   $x(l)$ to be convex. {This condition applies to various classic distributions, e.g., (1) the uniform distribution admits a $0.995$-approximation, (2) the exponential distribution can have the exact optimal mechanism, (3) the standard Weibull distribution (truncated on $[0, 1]$) with parameter $\gamma \ge 2$ (specifically we let $\gamma=  10$) admits approximation ratio $\approx 0.75$.} In some cases, it may be difficult to write down the explicit formula of $x(l)$ and verify its convexity. We find it relatively easy to compute $l(x)$, the inverse of $x(l)$. Hence, to verify the convexity of $x(l)$, we only need to verify the concavity of $l(x)$. An example of this idea is in Appendix \ref{example_examplecomputexlandinitialnumber}.

Before presenting the proof of Proposition \ref{theroemboudningapproximation}, we show some properties for convex $x(l)$. 

\begin{observation}\label{intersectionpointobservation}
The function $x = x(l)^k$, for any integer $k >1$ and $l \in [0, 1]$, intersects with line $x = l$ at most twice: one point intersecting is $(1, 1)$, and the other one (if exists) is $(l[k], l[k])$ (called intersection point as in Definition \ref{definitioninitiaointersection}). Furthermore, when $k > K_0$ increases, $l[k]$ decreases and $\lim_{k\to \infty} l[k] = 0$.
\end{observation}

The idea of Observation \ref{intersectionpointobservation} is depicted in Figure \ref{figobservationaaaproofA}. Given an initial number $K_0$ and $K = \lfloor \log_{x(l')} l' \rfloor$ for some $l'$, we can observe from the figure that: i) If $k \le K_0$, the curve $x(l)^k$ will be above line $x=l$; ii) If $K_0\le k \le K$,  then $l'\le l[K]$ and $x(l')^k$ will be above point $(l', l')$. If $l' = l[K]$, then there are exactly $K$ signals constructed within $[l', 1]$; and iii) If $k > K$, $x(l')^k$ will be below $(l', l')$. 


The following key lemma characterizes the monotonicity of convex $x(l)$.

\begin{lemma}\label{convexmontonefunction}
$x(l)$ is a monotone increasing function.
\end{lemma}

Now we are ready to show the proof of Proposition \ref{theroemboudningapproximation}. The high-level idea of the proof is by upper bounding the probability mass $z$ of sending (at most one) \emph{non-optimal} signal pair $(\sigma_0, \sigma_0)$, because when sending other signal pairs $(\sigma_i, \sigma_{i+1})$ as constructed in the signaling scheme, the signaling scheme achieves the maximum of expected revenue expressed in (\ref{expectedrevenuegivenrs}).

\begin{proofof}{ Proposition \ref{theroemboudningapproximation}}
Without loss of generality, we consider the case where there is only one pair of CTR vectors, $(l, 1)$ and $(1, l)$ such that $\lambda((l, 1)) = \lambda(1, l)) = \frac{1}{2}$. By the definition of symmetric environments, the analysis can be easily generalized to the case of more than two CTR vectors.

The following lemma segments $[0, 1]$ by the intersection points $l[k]$'s, and shows how the probability mass $z$ changes within each segment. Note that starting from $k = K_0 + 1$, $x(l)^k$ crosses the line $x=l$ (see Figure \ref{figobservationaaaproofB}). Then, we define $l[K_0]=1$.


\begin{figure} 
\begin{subfigure}[t]{0.5\textwidth}
\centering
\includegraphics[width=0.795\textwidth]{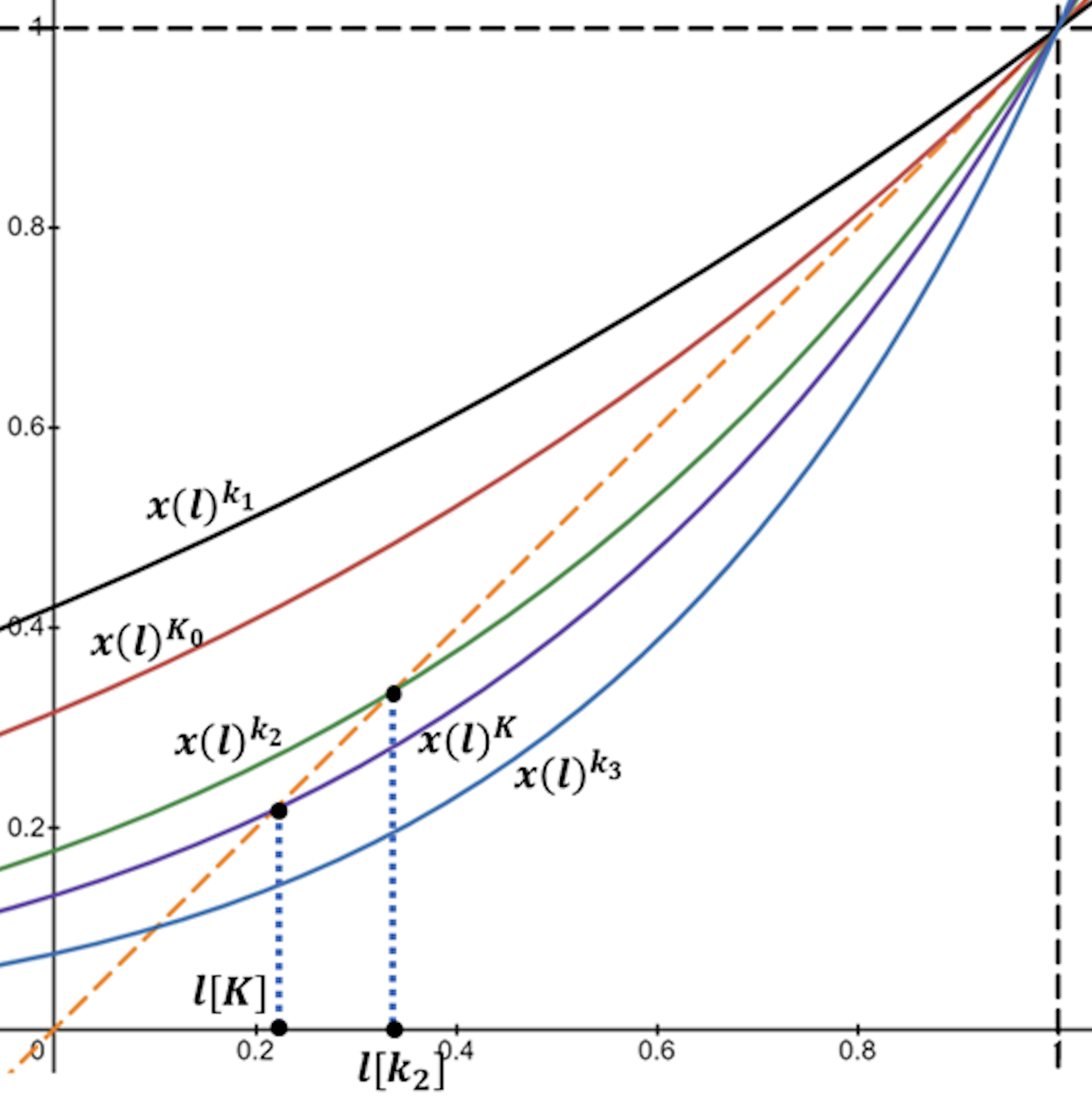} \caption{}
\label{figobservationaaaproofA}
\end{subfigure}
\hfill 
\begin{subfigure}[t]{0.5\textwidth}
\centering
\includegraphics[width=0.78\textwidth]{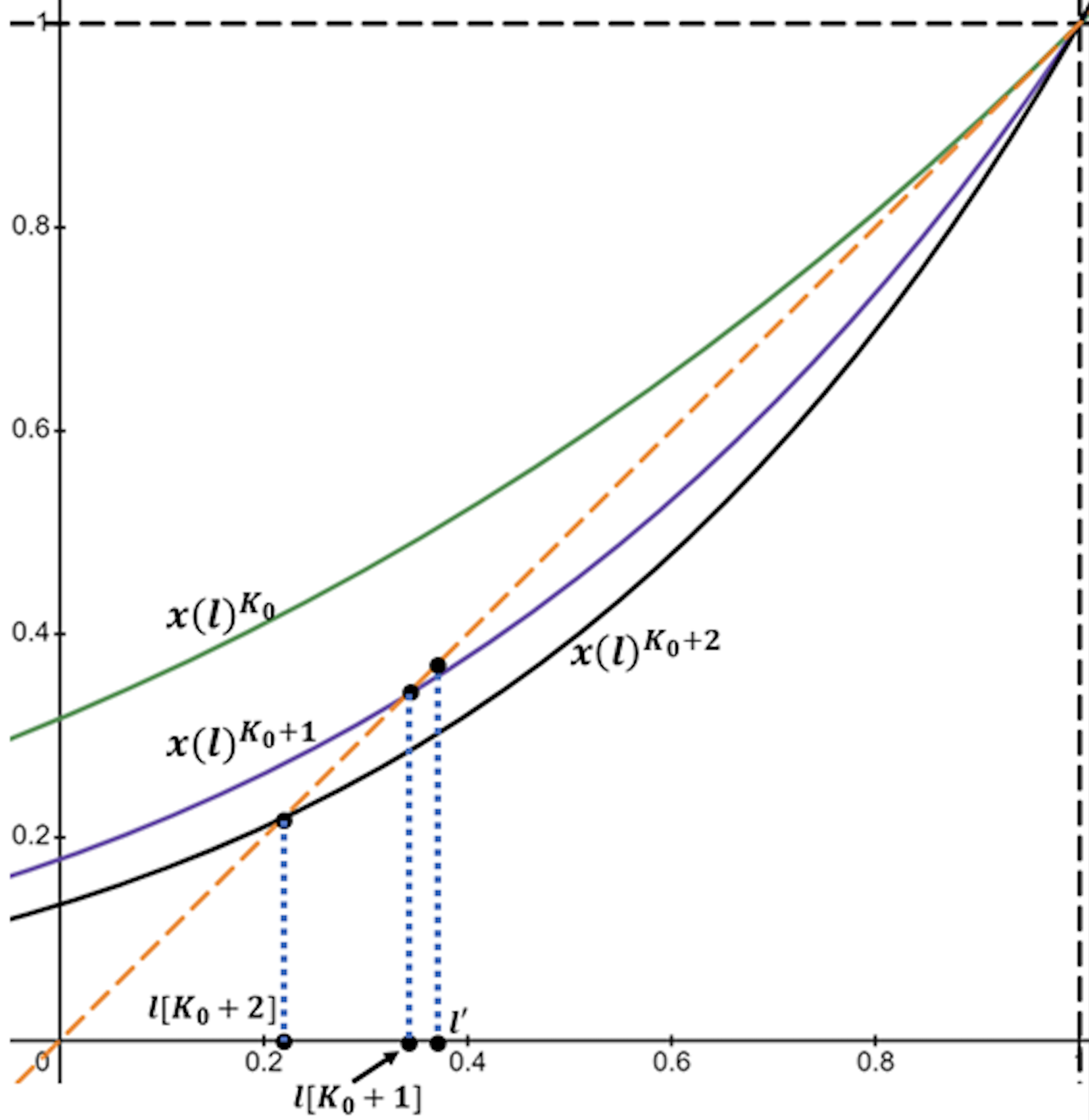} \caption{}\label{figobservationaaaproofB}
\end{subfigure}
\caption{$K_0$ is initial number. $l[k]$ is the intersection point. $k_1 < K_0 < k_2<K <k_3$. $x(l)$ is convex.}
\end{figure}


\begin{lemma}\label{opensetmonotonictiy}
Given $k \ge K_0 \ge 1$, the probability mass $z > 0$ is decreasing in $l \in (l[k+1], l[k]]$.
\end{lemma}
The proof of Lemma \ref{opensetmonotonictiy} is in Appendix \ref{app_opensetmonotonictiy}. The reason we only care about $(l[k+1], l[k]]$ instead of the closed interval $[l[k+1], l[k]]$ is that i) for $l[k+1] <l \le l[k]$, there will be  $k$ signals constructed in the interval $[l, 1]$, and ii) if $l = l[k+1]$, 
the seller can construct a signaling scheme with a list of $k+1$ signals where all the signal pairs sent have the optimal signal ratio $x(l[k+1])$ and the probability mass $z = 0$.

One implication of Lemma \ref{opensetmonotonictiy} is that when there are at most $k$ signals constructed, the probability mass $z$ achieves its maximum when $l$ approaches $l[k+1]$ from above, i.e., $l \downarrow l[k+1]$. Therefore, at the limit $ l[k+1]$, the probability mass $z$ is computed with  $S(k, l[k+1])$ where 
$\sigma_0 = x(l[k+1])^k, \sigma_1 = x(l[k+1])^{k-1}, \cdots, \sigma_i = x(l[k+1])^{k-i}$. 
In another word,  when $l = l[k+1]$, instead of constructing a scheme with a list of $k+1$ signals where all the signal pairs sent have signal ratio $x(l[k+1])$,  the seller constructs a signaling scheme with a list of $k$ signals starting from $\sigma_0 = x(l[k+1])^k$, which sends signal pair $(\sigma_{i-1}, \sigma_i)$ of signal ratio $x(l[k+1])$ and one signal pair $(\sigma_0, \sigma_0)$. 

By the above analysis, to upper bound $z$, we only need to compare its values at these limit points $l[k]$'s. Alternatively, we compare $S(k, l[k+1])$ for different $k$.
The following lemma shows that $S(k, l[k+1])$ is increasing in $k$, whose proof is in Appendix \ref{proof_of_lemma11finalsection}.
\begin{lemma}\label{lemmaskdecreasing}
Given  $k-1 \ge K_0 \ge 1$, we have $S(k, l[k+1])\ge S(k-1, l[k])$.
\end{lemma}

\medskip
Lemma \ref{lemmaskdecreasing} implies that  $S(K_0, l[K_0 + 1])$ is the minimum compared with other $S(k, l[k+1])$ for $k > K_0$. The quantity $S(K_0, l[K_0 + 1])$ is computed as (\ref{equantitysk}) with 
\[\sigma_0 = x(l[K_0+1])^{K_0}, \sigma_1 = x(l[K_0+1])^{K_0-1}, \dots, \sigma_i = x(l[K_0+1])^{K_0-i}\]
Since both CTR vectors $(1, l)$ and $(l, 1)$ will send signal $(\sigma_0, \sigma_0)$, the upper bound of the probability mass of sending signal pair $(\sigma_0, \sigma_0)$ in the constructed signaling scheme  is $z^* = 1/S(K_0, l[K_0 + 1])$. 

The above analysis generalizes to the case of more than two CTR vectors. Note that given any CTR vector $(h, l)$, in the symmetric environment, we can find one CTR vector $(l, h)$ with equal probability, i.e., $\lambda((h, l)) = \lambda((l, h))$. Hence, we can separately consider these pairs of CTR vectors when designing a signaling scheme and each of them achieves at least $(1-z^*)$ (multiplicatively) approximation. Therefore, the constructed signaling scheme can achieve at least $(1-z^*)$ (multiplicatively) approximation.
\end{proofof}

\section{Conclusions and Open Problems}
This paper studies the natural Bayesian variant of the calibrated click-through auction of \cite{bergemann2021calibrated}. We focus on the seller's information design problem to maximize the expected revenue. In general environments, we develop an FPTAS to compute an approximately optimal signaling scheme. 
In a symmetric environment, we give a simple and prior-free signaling scheme with a constant approximation guarantee for not-too-fluctuating value distributions. 

Our results raise many  interesting questions for future research. Below we discuss some of them. The first is to develop a simple signaling scheme for more than two bidders in symmetric environments. 
More generally, it is interesting to study more efficient algorithms for optimal signaling. Though our FPTAS for the general model enriches our understanding of the tractability of the problem,  such an algorithm may not be ideal from a practitioner's perspective. The second is to consider the worst-case approximation as in Proposition \ref{theroemboudningapproximation} but without the convexity assumption. In Appendix \ref{discussionconvexitymorethan2buyer}, we provide more discussions about these two open problems.



\bibliography{sample-base}

\appendix

\section{Additional Discussions on Related Works}\label{appendix_additional_related_works}

\noindent\textbf{Related Works on Information Design in Ad Auctions.} 
The most relevant literature to our work is information design in auctions. Bro Miltersen and Sheffet~\cite{bro2012send} and Emek et al.~\cite{emek2014signaling} study information design in a second price auction. They show that an optimal information revelation  policy can be obtained if the bidders' valuations are known, while the problem becomes \NP-hard if randomness is involved. Later,  Cheng et al.~\cite{cheng2015mixture} prove a PTAS for the Bayesian model of \cite{emek2014signaling}. Fu et al.~\cite{fu2012ad} study information design in a Myerson auction and prove that full revelation is optimal. Badanidiyuru et al.~\cite{badanidiyuru2018targeting} study signaling scheme design in a second price auction motivated by re-targeting feature in per-impression ad auctions.
Recently, Bergemann et al.~\cite{bergemann2021calibrated} provided characterizations for the optimal information disclosure policy in a symmetric calibrated click-through auction with complete information of bidder values.

\noindent\textbf{Related Works on Sale of Information.} Another line of related literature is the sale of information, which also selectively reveals information for revenue improvement but without a fixed mechanism format like ours. As far as we know, Babaioff et al. ~\cite{babaioff2012optimal} is the first to study the  information selling problem computationally, whose results are later improved in \cite{chen2020selling} by directly recommending action as a signal. Bergemann et al.~\cite{bergemann2018design} characterize the optimal information pricing in a binary environment, after which Cai and Velegkas~\cite{cai2020sell} study the problem computationally. Liu et al.~\cite{liu2021optimal} characterize the optimal scheme for selling information to a binary-decision buyer. A series of works  along this line studies   information selling in varies  setups \cite{li2021selling,zheng2021optimal,bergemann2022selling}.

\noindent\textbf{Related Works on Bayesian Persuasion.} Our work is related to the general theme of Bayesian persuasion, in which the sender sends signals to influence the receivers' decision. Information design in auctions can be viewed as one application of Bayesian persuasion. This problem was first studied by Kamenica and Gentzkow~\cite{kamenica2011bayesian}. The most relevant literature is the private Bayesian persuasion, as we need to design a private signaling scheme for multiple receivers (i.e., the bidders). To the best of our knowledge, the private Bayesian persuasion was proposed by Arieli and Babichenko~\cite{arieli2019private}, followed by a series of recent algorithmic studies   \cite{babichenko2022multi,dughmi2017algorithmicnoex,babichenko2017algorithmic}.
There are also other recent works on public persuasion \cite{xu2020tractability} and persuasion with limited resources \cite{gradwohl2021algorithms}. We also refer interested readers to a comprehensive survey by Dughmi~\cite{dughmi2017algorithmic}.

\jjr{\noindent\textbf{Related Works on Automated Bidding in Auctions.} While our work is focused on manual bidding, it is an interesting direction to pursue how the information design would influence those automated learning agents in auctions. We discuss some recent advances in this topic. Balseiro and Gur~\cite{balseiro2019learning} studied the budget-constrained agent’s bidding learning problem in repeated auctions. They prove the asymptotic optimality of proposed adaptive pacing strategies. Mehta and Perlroth~\cite{mehta2023auctions} studied the effects of the lack of bidding strategy commitments in auto-bidding worlds. While the automated learning agent can represent the users in auctions, there is a \emph{meta-game} between users. Kolumbus and Nisan~\cite{kolumbus2022and} and Kolumbus and Nisan~\cite{kolumbus2022auctions} studied the incentive problem of users in this meta-game. Feng et al.~\cite{feng2023strategic} further studied this problem in an auto-bidding world. 
}

\section{Discussions on Pay-Per-Click Policy, Commitment Power and Truthfulness}\label{append_payperclick}
\jjr{
{\textbf{Pay-per-click Policy.}} The pay-per-click (PPC) policy is a widely used auction format by advertising platforms such as Taobao~\cite{zhu2017optimized}, Microsoft Bing~\cite{wikimedia2024payperclick}, and TikTok~\cite{search2024whatisppc}. A systematic survey of more than $50$ papers on PPC regarding its various properties, practices, and use cases can be found in~\cite{kapoor2016pay}. There are generally three ad auction formats: pay-per-impression, pay-per-click, and pay-per-conversion. While it is difficult to get data about the exact fraction that each of the above auction formats is used on each platform (and these numbers also change over time), it is safe to claim that PPC consists of a significant portion, with many billions of dollars of revenue.

\noindent{\textbf{Commitment Power.}} We assume the auctioneer can commit to both the auction mechanism and the information revelation policy~\cite{emek2014signaling,badanidiyuru2018targeting,bergemann2021calibrated}. These are standard assumptions in auction design and information design. They are also justified in our domain of interest since auctioneers like Microsoft and Amazon are often regarded as authority parties and thus have the commitment power. Both the auction mechanism and information revelation policy are implemented as software and explained in the user manual, which is also one format of commitment.

\noindent{\textbf{Truthfulness.}} While truthful bidding is a dominant strategy in the single-item second-price auction in our case, it has been shown that humans do not typically bid truthfully even when advised to do so \cite{noti2014experimental}, and there may exist other, non-truthful, Nash equilibria. To be more precise, we clarify that the buyer's truthful bidding strategy is a behavior assumption but typically adopted in literature~\cite{emek2014signaling,bergemann2021calibrated,badanidiyuru2018targeting}.
}

\section{Missing Proofs in Section \ref{sectiongeneralenvironmentdiscrete}}\label{appendixproofoftheorem31claim}

\subsection{Proof of Lemma \ref{continousfunctionRrs}} \label{Proof_of_Corollary_32constant}


\begin{proofof}{Lemma \ref{continousfunctionRrs}}
First notice that for any $s = (s_1, s_2)$ and $s' = (s_1', s_2')$ within $[\underline{r}, 1]$, we know that 
\begin{equation}\label{equationarraycontinuity}
    \vert\frac{s_2}{s_1} - \frac{s_2' }{s_1'}\vert \le |\frac{s_2}{s_1} - \frac{s_2}{s_1'}| + | \frac{s_2}{s_1'} - \frac{s_2' }{s_1'}| \le \frac{1}{\underline{r}^2} |s_1 -s_1'| + \frac{1}{\underline{r}} |s_2 -s_2'|\le C_s \cdot\Vert s -s' \Vert
\end{equation}
where $C_s = \frac{\sqrt{2}}{\underline{r}} + \frac{\sqrt{2}}{\underline{r}^2}$ is a constant.

Given a signal vector $s = [s_1, s_2, \dots, s_n] \in [\underline{r}, 1]^n$, we can pick any two bidders $i$ and $j$  with $v_is_i$ and $v_j s_j$ as two largest  values. Then, the expected revenue of $R(r, s)$ in (\ref{expectedrevenuegivenrsgeneral}) is generalized to
\begin{small}
\begin{align*}
\begin{split}
R(r, s) 
&= \int_v R(r, v, s) f(v)dv \\
&= \sum_{(s_i, s_j) \in s} \int_a^b (r_i \cdot \frac{v_j s_j}{s_i})\cdot \prob{v_is_i \ge v_js_j| v_j} \prod_{k \neq i, j} \prob{v_k s_k \le v_j s_j|v_j} f_j(v_j) dv_j \\
&\qquad + \int_a^b (r_j \cdot \frac{v_is_i}{s_j} )\prob{v_js_j \ge v_is_i|v_i} \prod_{k\neq i, j} \prob{v_ks_k \le v_is_i|v_i} f_i(v_i) dv_i
\end{split}
\end{align*} 
\end{small}where $\prob{v_is_i \ge v_js_j| v_j}\prod_{k \neq i, j} \prob{v_k s_k \le v_j s_j|v_j}$ is the probability that bidder $i$ wins the auction and bidder $j$ has the second largest product given $v_j$. Hence, the first integration within the summation computes the expected revenue for the case that bidder $i$ wins the auction while bidder $j$ gives the second largest product value $v_j s_j$. The second integration has a similar interpretation.

Since the sum of Lipschitz continuous functions is still Lipschitz continuous, we only need to show the Lipschitz continuity of
\[\int_a^b (r_i \cdot \frac{v_j s_j}{s_i})\cdot \prob{v_is_i \ge v_js_j| v_j} \prod_{k \neq i, j} \prob{v_k s_k \le v_j s_j|v_j} f_j(v_j) dv_j.\]
Without loss of generality, we assume the two bidders are $i = 1$ and $j = 2$. 
We can see that 
\begin{equation*}
\begin{split}
&\int_a^b(r_2 \cdot \frac{v_1 s_1}{s_2})\cdot \prob{v_2s_2 \ge v_1s_1| v_1} \prod_{k \neq 1, 2} \prob{v_k s_k \le v_1 s_1|v_1} f_1(v_1) dv_1\\
&= \int_a^{a\frac{s_2}{s_1}} (r_2 \cdot \frac{v_1 s_1}{s_2})\cdot \int_a^b f(v_2) dv_2 \prod_{k \neq 1, 2} \prob{v_k s_k \le v_1 s_1|v_1} f_1(v_1) dv_1 \\ 
&\quad \quad  + \int_{a\frac{s_2}{s_1}}^b (r_2 \cdot \frac{v_1 s_1}{s_2})\cdot \int_{v_1\frac{s_1}{s_2}}^b f(v_2) dv_2 \prod_{k \neq 1, 2} \prob{v_k s_k \le v_1 s_1|v_1} f_1(v_1) dv_1
\end{split}
\end{equation*}
We will only show Lipschitz continuity of 
\[\int_{a\frac{s_2}{s_1}}^b (r_2 \cdot \frac{v_1 s_1}{s_2})\cdot \int_{v_1\frac{s_1}{s_2}}^b f(v_2) dv_2 \prod_{k \neq 1, 2} \prob{v_k s_k \le v_1 s_1|v_1} f_1(v_1) dv_1\]
A similar analysis applies to other parts of $R(r, s)$.
Given any $s$ and $s'$, we have 
\begin{small}
\begin{align}
&g(s, s') \\
&= \Big| \int_{a\frac{s_2}{s_1}}^b (r_2 \cdot \frac{v_1 s_1}{s_2})\cdot \int_{v_1\frac{s_1}{s_2}}^b f(v_2) dv_2 \prod_{k \neq 1, 2} \prob{v_k s_k \le v_1 s_1|v_1} f_1(v_1) dv_1 \notag\\
&\qquad - \int_{a\frac{s_2'}{s_1'}}^b (r_2 \cdot \frac{v_1 s_1'}{s_2'})\cdot \int_{v_1\frac{s_1'}{s_2'}}^b f(v_2) dv_2 \prod_{k \neq 1, 2} \prob{v_k s_k' \le v_1 s_1'|v_1} f_1(v_1) dv_1 \Big|\notag\\ 
&= \Big|\int_{a\frac{s_2}{s_1}}^{a\frac{s_2'}{s_1'}} (r_2 \cdot \frac{v_1 s_1}{s_2})\cdot \int_{v_1\frac{s_1}{s_2}}^b f(v_2) dv_2 \prod_{k \neq 1, 2} \prob{v_k s_k \le v_1 s_1|v_1} f_1(v_1) dv_1 \Big| \notag\\
& \quad+ \Big| \int_{a\frac{s_2'}{s_1'}}^b r_2\cdot v_1 f_1(v_1) \Big( \frac{s_1}{s_2} (1- F(v_1\frac{s_1}{s_2})) \prod_{k\neq 1, 2} F_k(\frac{v_1s_1}{s_k})  - \frac{s_1'}{s_2'} (1- F(v_1\frac{s_1'}{s_2'})) \prod_{k\neq 1, 2} F_k(\frac{v_1s_1'}{s_k'}) \Big)  dv_1 \Big| \notag\\
&\le \frac{abf_1^{max}}{\underline{r}} C_s \Vert s- s'\Vert \\
& \quad \quad +  \int_{a\frac{s_2'}{s_1'}}^b r_2\cdot v_1 f_1(v_1) \Big| \frac{s_1}{s_2} (1- F_2(v_1\frac{s_1}{s_2})) \prod_{k\neq 1, 2} F_k(\frac{v_1s_1}{s_k})   - \frac{s_1'}{s_2'} (1- F_2(v_1\frac{s_1'}{s_2'})) \prod_{k\neq 1, 2} F_k(\frac{v_1s_1'}{s_k'}) \Big|  dv_1 \notag\\
&\le   \frac{abf_1^{max}}{\underline{r}} C_s \Vert s- s'\Vert   \notag\\
& \quad+\int_a^b r_2 v_1 f_1(v_1)\cdot \Big| \frac{s_1}{s_2} \prod_{k \neq 1, 2} F_k(\frac{v_1 s_1}{s_k})- \frac{s_1'}{s_2'} \prod_{k \neq 1, 2} F_k(\frac{v_1 s_1'}{s_k'}) \Big| dv_1 \label{nlipschitezcond1}\\
& \quad + \int_a^b r_2 v_1 f_1(v_1)\cdot \Big| \frac{s_1}{s_2} \prod_{k \neq 1} F_k(\frac{v_1 s_1}{s_k}) - \frac{s_1'}{s_2'} \prod_{k \neq 1} F_k(\frac{v_1 s_1'}{s_k'}) \Big| dv_1 \label{nlipschitezcond12}
\end{align}
\end{small}Since the sum of Lipschitz continuous functions is Lipschitz continuous, we only need to prove that there exists some constant $C_1$ such that  (\ref{nlipschitezcond1}) is upper bounded (similar analysis can be applied to (\ref{nlipschitezcond12})),
\begin{equation}\label{constantnbuyerslipschitz}
\int_a^b r_2 v_1 f_1(v_1)\cdot|\frac{s_1}{s_2} \prod_{k \neq 1, 2} F_k(\frac{v_1 s_1}{s_k}) - \frac{s_1'}{s_2'} \prod_{k \neq 1, 2} F_k(\frac{v_1 s_1'}{s_k'})| dv_1 \le C_1\cdot \Vert s-s' \Vert 
\end{equation}
 This can be done as follows. We firstly note that
\begin{small}
\begin{equation*}
\begin{split}
& \Big|\frac{s_1}{s_2} \prod_{k \neq 1, 2} F_k(\frac{v_1 s_1}{s_k}) - \frac{s_1'}{s_2'} \prod_{k \neq 1, 2} F_k(\frac{v_1 s_1'}{s_k'})\Big| \\
&\le \Big|\frac{s_1}{s_2} \prod_{k \neq 1, 2} F_k(\frac{v_1 s_1}{s_k}) - \frac{s_1'}{s_2'} \prod_{k \neq 1, 2} F_k(\frac{v_1 s_1}{s_k}) + \frac{s_1'}{s_2'} F_3(\frac{v_1 s_1}{s_3}) \prod_{k \neq 1, 2, 3} F_k(\frac{v_1 s_1}{s_k}) - \frac{s_1'}{s_2'} F_3(\frac{v_1 s_1'}{s_3'}) \prod_{k \neq 1, 2,3} F_k(\frac{v_1 s_1}{s_k})\\
& \quad + \frac{s_1'}{s_2'} F_3(\frac{v_1 s_1'}{s_3'}) F_4(\frac{v_1 s_1}{s_4}) \prod_{k \neq 1, 2,3, 4} F_k(\frac{v_1 s_1}{s_k})  - \frac{s_1'}{s_2'} F_3(\frac{v_1 s_1'}{s_3'}) F_4(\frac{v_1 s_1'}{s_4'}) \prod_{k \neq 1, 2,3, 4} F_k(\frac{v_1 s_1}{s_k}) + \cdots - \frac{s_1'}{s_2'} \prod_{k \neq 1, 2} F_k(\frac{v_1 s_1'}{s_k'})\Big| \\
&\le \Big|\frac{s_1}{s_2} \prod_{k \neq 1, 2} F_k(\frac{v_1 s_1}{s_k}) - \frac{s_1'}{s_2'} \prod_{k \neq 1, 2} F_k(\frac{v_1 s_1}{s_k})\Big| +\Big| \frac{s_1'}{s_2'} F_3(\frac{v_1 s_1}{s_3}) \prod_{k \neq 1, 2, 3} F_k(\frac{v_1 s_1}{s_k}) - \frac{s_1'}{s_2'} F_3(\frac{v_1 s_1'}{s_3'}) \prod_{k \neq 1, 2,3} F_k(\frac{v_1 s_1}{s_k})\Big| + \cdots
\end{split}
\end{equation*}
\end{small}Without loss of generality, we consider the first two terms in the above. We can see that 
\begin{align*}
|\frac{s_1}{s_2} \prod_{k \neq 1, 2} F_k(\frac{v_1 s_1}{s_k}) - \frac{s_1'}{s_2'} \prod_{k \neq 1, 2} F_k(\frac{v_1 s_1}{s_k})| &= |\frac{s_1}{s_2}  - \frac{s_1'}{s_2'} | \prod_{k \neq 1, 2} F_k(\frac{v_1 s_1}{s_k}) \\
&\le |\frac{s_1}{s_2}  - \frac{s_1'}{s_2'} |\\
&\le C_s \cdot \Vert s -s'\Vert
\end{align*}
where the last inequality is by (\ref{equationarraycontinuity}). Also, we can see
\begin{align*}
&| \frac{s_1'}{s_2'} F_3(\frac{v_1 s_1}{s_3}) \prod_{k \neq 1, 2, 3} F_k(\frac{v_1 s_1}{s_k}) - \frac{s_1'}{s_2'} F_3(\frac{v_1 s_1'}{s_3'}) \prod_{k \neq 1, 2,3} F_k(\frac{v_1 s_1}{s_k})| \\
&= | F_3(\frac{v_1 s_1}{s_3}) - F_3(\frac{v_1 s_1'}{s_3'}) | \frac{s_1'}{s_2'} \prod_{k \neq 1, 2,3} F_k(\frac{v_1 s_1}{s_k}) \\
&\le | F_3(\frac{v_1 s_1}{s_3}) - F_3(\frac{v_1 s_1'}{s_3'}) |\cdot \frac{1}{\underline{r}} \\
&\le f_3^{\max} \frac{v_1}{\underline{r}}  |\frac{s_1}{s_3} - \frac{s_1'}{s_3'}|\\
&\le f_3^{\max} \frac{v_1}{\underline{r}}  C_s \cdot \Vert s -s'\Vert
\end{align*}
$f^{\max}_i$ is the maximum of density function $f_i$ and $f^{\max}_i < \infty$ exists by the no-point-mass assumption. By similarly applying the inequality to other terms and $\E[v_1^2] < \infty$, we finally will have some constant $C_1 = E[v_1] C_s + \frac{E[v_1^2] C_s}{\underline{r}} \sum_{i=3}^n f_i^{\max}$ such that  (\ref{constantnbuyerslipschitz}) is satisfied. By assumption, we know $\E[v_i^2] < \infty$ and thus $\E [v_i] < \infty$ exist. Similar analysis can show that (\ref{nlipschitezcond12}) is also upper bounded with some constant. Therefore, there must exist some constant $C_n$ such that $|R(r, s) - R(r, s')| \le C_n \cdot \Vert s- s' \Vert$. 

The above proof assumes $v_i\in [a, b]$ with $b <\infty$. For the case of $b =\infty$, the proof is similar and thus is omitted.
\end{proofof}

\subsection{Proof of Theorem \ref{claimriunderline}}\label{appendix_proofof_theorem_section3main}

\begin{proofof}{Theorem \ref{claimriunderline}}
{We start by   overviewing the procedure  of our proof}. Given an optimal signaling scheme $\pi$ in the continuous signal space,  we first obtain a new signaling scheme by proportionally scaling down $\pi$ by $1-O(\epsilon)$. This step will save probability mass $O(\eps)$. Then, we round the signals in $\pi$ to some nearby discretized signals, which  may no longer be calibrated. Finally, we use the reserved probability mass $O(\eps)$ to adjust the obtained signaling scheme so that it becomes calibrated again. We show that the above discretization only results in (multiplicatively) $O(\eps)$ loss of the optimal revenue.
In the following, we use $2$-bidder case for easy exposition then generalize the results to $n$-bidder case. The probability mass is defined as $x(r, s) = \lambda(r)\pi(s|r)$ so that $\sum_{r, s} x(r, s) = 1$.

\medskip

\textbf{Step 1: Na\"ive Discretization of Signal Space.} \quad We show the discretization of the signal space for bidder $1$. A similar discretization is applied to bidder $2$. Let the smallest CTR of bidder $1$ be $\underline{r}_1$ and the largest one be $\overline{r}_1$. Since the signal $s_1$ sent by bidder $1$ must be within $[\underline{r}_1, \overline{r}_1]$, we discretize the signal space  $[\underline{r}_1, \overline{r}_1]$ by $\frac{1}{\eps}$ steps, where $\eps$ is some small positive number and the length of each step is $\rho = \eps(\overline{r}_1 - \underline{r}_1)$. Hence, the discrete signal space is $D = \{\underline{r}_1, \underline{r}_1+\rho, \underline{r}_1+2\rho, \dots, \underline{r}_1+k\rho, \frac{\overline{r}_1 + \underline{r}_1}{2}, \underline{r}_1+(k+1)\rho, \dots, \overline{r}_1 \}$. It is important to note that we include the middle point $\frac{\overline{r}_1 + \underline{r}_1}{2}$ in the discretized space.  To ease exposition, we assume the middle point $\frac{\overline{r}_1 + \underline{r}_1}{2}$ coincides with $\underline{r}_1 + m\rho$ for some $m$. 

\textbf{Step 2: Rounding while Preserving Calibration Constraints.} \quad In this step, we address the key challenge of proving our algorithm: rounding signals when maintaining the calibration constraints. We divide this step into three small steps.

\textbf{Step 2a: Proportionally Scaling Down $\pi$.} \quad Given the optimal signaling scheme $\pi$ in the original continuous signal space and the optimal revenue $OPT$, we first construct a  new signaling scheme $\pi'$ by proportionally scaling down the conditional probability $\pi'(s|r) = (1-O(\eps))\pi(s|r)$. Then, at most $O(\eps)$ probability mass is reserved. Also, $O(\eps)\cdot OPT$ revenue is lost by comparing $\pi'$ with $\pi$ because the objective function is linear in $\pi$. The calibration constraint is still maintained in $\pi'$.

\textbf{Step 2b: Rounding Continuous Signals To   Nearby Discrete Signals.} \quad We round the continuous signals to a discretized signal space, by which the calibration constraint may be violated. We retain the calibration constraint by carefully adjusting the signaling scheme.  Denote the constructed signaling scheme obtained in this step as $\overline{\pi}$. For ease of notation, we reuse $\pi$ to denote $\pi'$ obtained in Step 2a. 

Without loss of generality, we consider bidder $1$ in the following discussion. 
The rounding process works as follows: If $s_1 > \frac{\overline{r}_1 + \underline{r}_1}{2}$ and $\underline{r}_1 + k \rho \le s_1 < \underline{r}_1 + (k+1)\rho$ for some $k$, then we round it to $\overline{s}_1 = \underline{r}_1 + k \rho$. If $s_1 \le \frac{\overline{r}_1 + \underline{r}_1}{2}$ and $\underline{r}_1+ k \rho < s_1 \le \underline{r}_1 + (k+1)\rho$ for some $k$, then we round it to $\overline{s}_1 = \underline{r}_1 + (k +1)\rho$. A similar rounding process is used for signal $s_2$. Let $\overline{\pi} (\overline{s}|r) = \pi(s|r)$.
By Lemma \ref{continousfunctionRrs}, we know that the rounding process results in loss $|R(r, s) - R(r, \overline{s})| \le O(\rho) \le O(\eps)$, due to the  Lipschitz continuity of $R(r, s)$. Hence, the rounding process causes at most $O(\eps)$ revenue loss. {Note that the signals within intervals $(\underline{r}_1+k\rho, \frac{\overline{r}_1 + \underline{r}_1}{2}]$ and $[\frac{\overline{r}_1 + \underline{r}_1}{2}, \underline{r}_1+(k+1)\rho)$ will be rounded to $\frac{\overline{r}_1 + \underline{r}_1}{2}$. We can first replace them with a new signal $s'$ such that the calibration constraint holds, then apply the same adjusting process as described below.}



\textbf{Step 2c: Adjust Signaling Scheme To Restore Calibration.} \quad Recall that in Step 2a, $O(\eps)$ probability mass is reserved. After the above rounding process, the calibration constraint may no longer hold. In the following, the reserved $O(\eps)$ probability mass will be used to make the signals calibrated again.

Suppose that the original signal $s_1 > \frac{\overline{r}_1 + \underline{r}_1}{2}$. When observing its rounded signal $\overline{s}_1 = \underline{r}_1 + k\rho$ for some $k$,  bidder $1$ computes the estimation $\mu = \E[r_1 | \overline{s}_1 = \underline{r}_1 + k\rho]$. We know that $\underline{r}_1 + k \rho \le \mu < \underline{r}_1 + (k+1)\rho$, whose argument is as follows. Suppose there are only two different signals $s_1^a$ and $s_1^b$ which will be rounded to $\overline{s}_1=\underline{r}_1 + k \rho$. With the signaling scheme $\pi$ (obtained after Step 2a), we have (the equality is due to that signals $s_1^a$ and $s_1^b$ are calibrated in $\pi$)
\[\underline{r}_1 + k \rho \le s_1^a =  \frac{\sum_r r_1 \lambda(r) \sum_{s: s_1 = s_1^a} \pi(s|r)}{\sum_r \lambda(r) \sum_{s: s_1 = s_1^a} \pi(s|r)} < \underline{r}_1 + (k+1)\rho \]
\[\underline{r}_1 + k \rho \le s_1^b =  \frac{\sum_r r_1 \lambda(r) \sum_{s: s_1 = s_1^b} \pi(s|r)}{\sum_r \lambda(r) \sum_{s: s_1 = s_1^a} \pi(s|r)} < \underline{r}_1 + (k+1)\rho \]
After the above rounding process, we have (note that by now, $\overline{\pi} (\overline{s}|r) = \pi(s|r)$ still holds.) 
\begin{equation*}
\begin{split}
&\underline{r}_1 + k \rho \le \mu = \E[r_1 | \overline{s}_1 = \underline{r}_1 + k\rho] \\
&=    \frac{\sum_r r_1 \lambda(r) \sum_{s: s_1 = s_1^b} \pi(s|r) + \sum_r r_1 \lambda(r) \sum_{s: s_1 = s_1^a} \pi(s|r) }{\sum_r \lambda(r) \sum_{s: s_1 = s_1^a} \pi(s|r) + \sum_r \lambda(r) \sum_{s: s_1 = s_1^a} \pi(s|r) } < \underline{r}_1 + (k+1)\rho.
\end{split}
\end{equation*}
The above argument can be extended to the case that multiple signals $s_1'$ are rounded to $\overline{s}_1 = \underline{r}_1 + k\rho$. 
For simplicity, we denote the posterior estimation as $\mu =\frac{b_k}{a_k}$, where $b_k = \E[r_1, \overline{s}_1 = \underline{r}_1 + k\rho] = \sum_r r_1 \lambda(r) \sum_{\overline{s}: \overline{s}_1 = \underline{r}_1+ k \rho} \overline{\pi}(\overline{s}|r)$ is the expected value of click-through-rate $r_1$ when the observed  signal is $\overline{s}_1$ and $ a_k =  \prob{\overline{s}_1 = \underline{r}_1 + k\rho} = \sum_{r} \lambda(r) \sum_{\overline{s}: \overline{s}_1 = \underline{r}_1 + k\rho} \overline{\pi}(\overline{s}|r)$ is the total probability mass of observing $\overline{s}_1$.  

If the calibration constraint does not hold, i.e., $\mu \neq \underline{r}_1 + k\rho$, then we can adjust the signaling scheme by picking some CTR state $(\underline{r}_1, r_2)$ for some arbitrary $r_2$, and map it to a new signal $(\underline{r}_1 + k \rho, r_2)$ with carefully tuned probability mass $\delta_k = \lambda(r = (\underline{r}_1, r_2)) \cdot \overline{\pi}(s = (\underline{r}_1 + k \rho, r_2)| r = (\underline{r}_1, r_2))$, i.e., {we add a new column to $\overline{\pi}$ with entry $\overline{\pi}(s = (\underline{r}_1 + k \rho, r_2)| r = (\underline{r}_1, r_2)) > 0$ and other entries 
in this column being $0$
to make the calibration constraint satisfied again}. Next, we argue that this step only needs $O(\eps)$ probability mass in total to make all the rounded signals $\overline{s}_1$ sent to bidder $1$ calibrated again, i.e., $\sum_k \delta_k = O(\eps)$. Note that we can send $(\underline{r}_1, r_2)$ for multiple $r_2$, and the analysis remains the same. 

Suppose that after mapping $(\underline{r}_1, r_2)$ to a new signal with probability mass $\delta_k$, the calibration constraint is satisfied, i.e., $\mu = \underline{r}_1 + k\rho$. Then, the posterior estimation can be computed as  $\mu = \frac{b_k +\underline{r}_1 \delta_k}{a_k + \delta_k} = \underline{r}_1 + k\rho$. Simple calculation leads to $\delta_k = \frac{b_k - a_k(\underline{r}_1 + k\rho)}{k\rho}$. 
Note that from the above discussion $\underline{r}_1 + k \rho \le  \frac{b_k}{a_k} < \underline{r}_1 + (k+1)\rho$, by which we know
\[b_k \ge a_k(\underline{r}_1 + k\rho), \quad b_k < a_k(\underline{r}_1 + (k+1)\rho),\]
Hence, we have $0<\delta_k < \frac{a_k\rho}{k\rho}$, since $b_k < a_k (\underline{r}_1 + (k+1)\rho)$.  Our choice of $s_1$ satisfies $ s_1 \ge  \underline{r}_1 + k \rho \ge  \frac{\overline{r}_1 + \underline{r}_1}{2}$. This implies $k\rho >  \frac{\overline{r}_1 - \underline{r}_1}{2}$.   Consequently,  $\delta_k < \frac{a_k\rho}{k\rho} \le 2a_k\eps$, because $\rho = \eps(\overline{r}_1 - \underline{r}_1)$. The set of signals used in the new signaling scheme for bidder $1$ will be $D \union \mathcal{R}$. 

Similar analysis as above is also applied to the case $s_1 \le \frac{\overline{r}_1 + \underline{r}_1}{2}$ with $\underline{r}_1 + k \rho < s_1 \le \underline{r}_1 + (k+1)\rho$. It will be rounded to some $\underline{r}_1 + (k+1)\rho$, and we similarly have $\underline{r}_1 + k\rho <\E[r_1 | \overline{s}_1 = \underline{r}_1 + (k+1)\rho] = \frac{b_{k+1}}{a_{k+1}} \le \underline{r}_1 + (k+1)\rho$. The $a_{k+1}$ and $b_{k+1}$  are defined similarly as in the case $s_1 > \frac{\overline{r}_1 + \underline{r}_1}{2}$.  To retain the calibration constraint, we pick the CTR state $r = (\overline{r}_1, r_2)$ for some arbitrary $r_2$ and map it to signal $(\underline{r}_1 + (k+1) \rho, r_2)$ with probability mass
\[\delta_{k+1} = \frac{a_{k+1}(\underline{r}_1 + (k+1)\rho) -b_{k+1}}{\overline{r}_1 - (\underline{r}_1 + (k+1)\rho)} \le \frac{a_{k+1}\rho}{\overline{r}_1 - (\underline{r}_1 + (k+1)\rho)} \le 2a_{k+1}\eps,\]
where the first inequality is due to $b_{k+1} > a_{k+1}(\underline{r}_1 + k\rho)$ while the second one is due to $s_1 \le \underline{r}_1 + (k+1)\rho \le \frac{\overline{r}_1 + \underline{r}_1}{2}$ in this case and thus $\overline{r}_1 - (\underline{r}_1 + (k+1)\rho) \ge \frac{\overline{r}_1-\underline{r}_1}{2}$. 

 We sum up the probability mass needed to adjust all the signals $\sum_k \delta_k \le 2\sum_k a_k \eps$ where $\sum_k a_k = \sum_k \prob{\overline{s}_1 = \underline{r}_1 + k\rho}$. Note that  $\sum_k a_k \le 1$ because of the probability constraint of a signaling scheme. Hence, $\sum_k \delta_k \le O(\eps)$. Similar rounding and adjustment process can be applied to  signal $s_2$ sent to bidder $2$ so that the calibration constraint is satisfied again. Hence, the adjustment needs total probability mass $O(\eps)$.

\textbf{Step 3: Preserve Probability Distribution Constraints.} \quad To make the new signaling scheme satisfy the probability distribution constraint, directly map the CTR $r$ to some new signal $s = r$ with the remaining probability mass after the adjustment of the signaling scheme.

\medskip
Finally, as discussed  above, the revenue loss is at most $O(\eps) OPT$ which is caused by the reservation of probability mass in Step 2a and the rounding in Step 2b. Step 2c and Step 3 will only increase revenue since we add additional signals to the scheme. The proof for the existence of an approximately optimal signaling scheme in the discrete signal space is then finished.

\textbf{Generalize to $n$ Bidders.}\quad Notice that from the above analysis, rounding and adjusting signal $s_i$ for bidder $i$ does not affect the calibration of signals of other bidders $j\neq i$. Hence, we can do the rounding and adjustments for each bidder separately. With a  finer discretization, i.e., the length of step $\rho_i=\frac{\eps}{n}(\overline{r}_i - \underline{r}_i)$, we can obtain an approximate scheme by applying the above analysis. The total probability mass used for adjustment is still $O(\eps)$. 

\textbf{The Time Complexity.}\quad  The number of discrete signals used is $O((\frac{n}{\eps})^n+|\mathcal{R}|^n)$.  There are at most $|\mathcal{R}|^n$  CTR states. Hence, there are $O((\frac{n|\mathcal{R}|}{\eps})^n+|\mathcal{R}|^{2n})$ unknown variables in total. With the set of discrete signals provided,  (\ref{ctrsecondproblem}) becomes a linear program. The number of constraints is then  also polynomial $poly(|R|^n, (\frac{n}{\eps})^n)$. Hence, we can compute the approximately optimal signaling scheme $\overline{\pi}$ in $poly(|\mathcal{R}|^n, (\frac{n}{\eps})^n)$ time by solving a linear program of (\ref{ctrsecondproblem}) with discretized signals. Note that if $n$ is some constant, our algorithm has poly-time complexity.
\end{proofof}


\subsection{Proof of Corollary \ref{corollaryclickthroughreservedprice}}\label{Proof_of_Corollary_corollaryclickthroughreservedprice}
\begin{proofof}{Corollary \ref{corollaryclickthroughreservedprice}}
Assume that there is some reserved price $p$ for the mechanism.  We assume that $p\in [a, b)$, which is usually a small value close to $a$. For the case of $p \le a$, the discussion is similar. Consider the case of $n=2$ bidders for easy exposition. The seller's expected revenue is reformulated as 
\begin{small}
\begin{align*}
   R(r, s) &= \int_a^b (r_1 \cdot \max\{\frac{v_2 s_2}{s_1}, p\})\prob{v_1s_1 \ge v_2s_2 \wedge v_1 \ge p| v_2} f_2(v_2) dv_2 \\
   &\quad +\int_a^b (r_2 \cdot \max\{\frac{v_1s_1}{s_2}, p \})\prob{v_2s_2 \ge v_1s_1 \wedge v_2 \ge p|v_1} f_1(v_1) dv_1\\
   &= \int_a^b (r_1 \cdot \max\{\frac{v_2 s_2}{s_1}, p\}) \int_{\max\{p, \frac{v_2s_2}{s_1}\}}^b f_1(v_1) dv_1  f_2(v_2) dv_2 \\
   &\quad + \int_a^{b} (r_2 \cdot \max\{\frac{v_1s_1}{s_2}, p \}) \int_{\max\{\frac{v_1s_1}{s_2}, p\}}^b f_2(v_2) dv_2   f_1(v_1) dv_1
\end{align*}   
\end{small}Note that in the proof of Theorem \ref{claimriunderline}, we only require $R(r, s)$ to be Lipschitz continuous. Therefore, if we can prove that $R(r, s)$ is Lipschitz continuous, then we can still have FPTAS to compute a signaling scheme achieving $1-O(\epsilon)$ approximation.

  In this proof, we only show Lipschitz continuity of 
  \[\int_a^{b} (r_2 \cdot \max\{\frac{v_1s_1}{s_2}, p \}) \int_{\max\{\frac{v_1s_1}{s_2}, p\}}^b f_2(v_2) dv_2   f_1(v_1) dv_1.\] 
  Similar proof applies to the Lipschitz continuity of 
  \[\int_a^b (r_1 \cdot \max\{\frac{v_2 s_2}{s_1}, p\}) \int_{\max\{p, \frac{v_2s_2}{v_1}\}}^b f_1(v_1) dv_1  f_2(v_2) dv_2.\] 
  We have 
  \begin{align*}
      g(s, s') 
      &= |\int_a^{b} (r_2 \cdot \max\{\frac{v_1s_1}{s_2}, p \}) \int_{\max\{\frac{v_1s_1}{s_2}, p\}}^b f_2(v_2) dv_2   f_1(v_1) dv_1\\
      & \quad- \int_a^{b} (r_2 \cdot \max\{\frac{v_1s_1'}{s_2'}, p \}) \int_{\max\{\frac{v_1s_1'}{s_2'}, p\}}^b f_2(v_2) dv_2   f_1(v_1) dv_1|\\
      &\le |\int_a^{b} (r_2 \cdot \max\{\frac{v_1s_1}{s_2}, p \})  \Big(\int_{\max\{\frac{v_1s_1}{s_2}, p\}}^b f_2(v_2) dv_2 - \int_{\max\{\frac{v_1s_1'}{s_2'}, p\}}^b f_2(v_2) dv_2 \Big)  f_1(v_1) dv_1|\\
      &\quad + |\int_a^{b} r_2 \cdot \Big( \max\{\frac{v_1s_1}{s_2}, p \} - \max\{\frac{v_1s_1'}{s_2'}, p \} \Big)  \cdot \int_{\max\{\frac{v_1s_1'}{s_2'}, p\}}^b f_2(v_2) dv_2  f_1(v_1)dv_1|\\
      &\le\int_a^{b} (r_2 \cdot \max\{\frac{v_1s_1}{s_2}, p \}) \Big(f_2^{\max} |\frac{v_1s_1'}{s_2'} - \frac{v_1s_1}{s_2}| \Big)  f_1(v_1) dv_1\\
      &\qquad  + \int_a^{b} r_2 \cdot | \frac{v_1s_1}{s_2} - \frac{v_1s_1'}{s_2'} | \int_{\max\{\frac{v_1s_1'}{s_2'}, p\}}^b f_2(v_2) dv_2  f_1(v_1)
  \end{align*}
By simple calculation as in Lemma \ref{continousfunctionRrs}, we have 
\[g(s, s') \le \frac{E[v_1]bf_2^{\max}}{\underline{r}}C_s \Vert s - s'\Vert + E[v_1] C_s \Vert s- s' \Vert. \]
Hence, the constant is $C_2 =  \frac{E[v_1]bf_2^{\max}}{\underline{r}}C_s + E[v_1] C_s$.  Similarly, we can derive the constant $C_1$ for 
\[\int_a^b (r_1 \cdot \max\{\frac{v_2 s_2}{s_1}, p\}) \int_{\max\{p, \frac{v_2s_2}{s_1}\}}^b f_1(v_1) dv_1  f_2(v_2) dv_2.\] 
Therefore, $R(r, s)$ is Lipschitz continuous with constant $C = C_1 + C_2$. The above analysis easily generalizes to any constant number of bidders by generalizing $R(r, s)$ as Lemma \ref{continousfunctionRrs}.
\end{proofof}

\section{Discussion on No-Point-Mass Assumption}\label{discussiononllipscontinityassumption}

The \emph{no-point-mass  assumption} is needed because i) the revenue for any fixed values as a function of signals is discontinuous around the tie-breaking points, and ii) the rounding construction does not preserve ties between the original scheme and the discretized scheme. Without the assumption, the optimal scheme may have a non-vanishing probability mass leading to ties, and our rounding may result in a revenue loss not vanishing as $\epsilon \rightarrow0$. 
We construct an example to explain as follows.
\begin{example}
    Let CTR $r=(1.0, 0)$ and signal $s=(0.81, 0.8)$. The tie-breaking rule: mechanism favors bidder $2$ when $v_1s_1=v_2s_2$. The value distribution is a point-mass distribution in $v_1=v_2=1$, i.e., $f_1(v_1=1)=f_2(v_2=1)=\infty$ and $f_1(v_1\neq 1)=f_2(v_2\neq 1)=0$, which is not Lipschitz continuous. With signal $s$, the winner is bidder $1$ with expected revenue $R(r, s)\approx 1$. If discretizing with step=0.1, $s$ may be rounded to $s’=(0.8, 0.8)$ or $s’=(0.9, 0.9)$, where the winner becomes bidder $2$ due to tie-breaking rule and $R(r,s’)=0$. The change of winner results in a severe loss of revenue and finally will lead to a large gap between continuous and discrete solutions.
\end{example}
Our rounding construction does not preserve ties due to the need to preserve the calibration constraints. As the revenue of fixed values is continuous in signals away from the tie-breaking points, the no-point-mass assumption is in fact implied by the probability mass of ties being zero for any signal profile, 
which is a mild assumption.

\section{Properties of the Optimal Signal Ratio} \label{app_propoertyofoptimalsignal}
The following lemma shows that given that the density function $f(v)$ is MHR over some bounded interval $[0, c]$, the optimal signal ratio $x^* = \frac{s_2}{s_1} \le 1$, which helps simplify the construction of the signaling scheme.
\begin{lemma}\label{lemmas2s1lessthan1}
Given a CTR vector $r = (1, l)$ for two bidders and the signal $(s_1, s_2)$, the expected revenue  (\ref{expectedrevenuegivenrs}) is maximized at the optimal signal ratio $x^*$, where $0 \le x^* = \frac{s_2^*}{s_1^*} \le 1$.
\end{lemma}
\begin{proof}
If the density function $f(v)$ is  MHR over some bounded interval $[0, c]$, then by letting $x=\frac{s_2}{s_1} \le 1$, we can explicitly write (\ref{expectedrevenuegivenrs}) as
\begin{align*}
    g(x)& = 1\cdot \int_0^c v_2 x\int_{v_2x}^c f(v_1)dv_1 f(v_2) dv_2  + l \cdot \int_0^{cx}  \frac{v_1}{x} \int_{\frac{v_1}{x}}^c f(v_2)dv_2 f(v_1) dv_1;
\end{align*}
while if $\frac{s_2}{s_1} \ge 1$, then by letting $y = \frac{s_1}{s_2} \le 1$, we can write (\ref{expectedrevenuegivenrs}) as
\begin{align*}
    h(y)& = 1 \cdot \int_0^c  \frac{v_2}{y} \int_{\frac{v_2}{y}}^c f(v_1)dv_1 f(v_2) dv_2   + l \cdot \int_0^{\frac{c}{y}}  v_1 y \int_{v_1y}^c f(v_2)dv_2 f(v_1) dv_1.\\
    &= 1 \cdot \int_0^{cy}  \frac{v_2}{y} \int_{\frac{v_2}{y}}^c f(v_1)dv_1 f(v_2) dv_2  + l \cdot \int_0^{c}  v_1 y \int_{v_1y}^c f(v_2)dv_2 f(v_1) dv_1.
\end{align*}
To prove expected revenue  (\ref{expectedrevenuegivenrs}) is maximized at $\frac{s_2^*}{s_1^*} \le 1$, it is equivalent to showing that $\max_{x\in[0, 1]} g(x) \ge \max_{y \in[0, 1]} h(y)$. Note that $v_1$ and $v_2$ are symmetric. Therefore, if we can prove that for any $x\in[0, 1]$, 
\begin{equation}\label{eq14lhsrhsgxhy}
\int_0^c v_2 x\int_{v_2x}^c f(v_1)dv_1 f(v_2) dv_2 \ge \int_0^{cx}  \frac{v_2}{x} \int_{\frac{v_2}{x}}^c f(v_1)dv_1 f(v_2) dv_2,
\end{equation}
then due to $l \le 1$, it implies that for any $x \in [0, 1]$, we have $g(x) \ge h(x)$, which further implies that $\max_{x\in[0, 1]} g(x) \ge \max_{y \in[0, 1]} h(y)$. 

Let $t = v_2 x$. Then
$0 \le t \le cx$ and $dt = x dv_2$. The left hand side of (\ref{eq14lhsrhsgxhy}) equals 
\begin{equation} \label{symbol14to15}
\int_0^c v_2 x\int_{v_2x}^c f(v_1)dv_1 f(v_2) dv_2
=\int_{0}^{cx} t \int_{t}^{c} f(v_1) dv_1 f(\frac{t}{x}) \frac{1}{x} dt
\end{equation}
For easy exposition, we replace $t$ in (\ref{symbol14to15}) with $v_2$, and (\ref{eq14lhsrhsgxhy}) can be rewritten as
\[\int_{0}^{cx} \frac{v_2}{x} \int_{v_2}^{c} f(v_1) dv_1 f(\frac{v_2}{x}) dv_2\ge \int_0^{cx}  \frac{v_2}{x} \int_{\frac{v_2}{x}}^c f(v_1)dv_1 f(v_2) dv_2\]
To show the above inequality holds for any $x$, 
notice that given any $v_2 \ge 0$, we have

\[ \frac{v_2}{x} \int_{v_2}^c f(v_1) dv_1 f(\frac{v_2}{x})   \ge   \frac{v_2}{x} \int_{\frac{v_2}{x}}^c f(v_1)dv_1 f(v_2), \]
which is equivalent to 
\[\int_{v_2}^c f(v_1) dv_1 f(\frac{v_2}{x})   \ge   \int_{\frac{v_2}{x}}^c f(v_1)dv_1 f(v_2) \]
and then
\[ (1 - F(v_2)) f(\frac{v_2}{x})   \ge   (1 - F(\frac{v_2}{x})) f(v_2) \]
Hence, we need to show 
\[ \frac{(1 - F(v_2))}{f(v_2)}    \ge  \frac{(1 - F(\frac{v_2}{x}))}{f(\frac{v_2}{x})} \]
which is equivalent to showing
\begin{align}  
\frac{f(v_2)}{(1 - F(v_2))}    \le  \frac{f(\frac{v_2}{x})}{(1 - F(\frac{v_2}{x}))} \label{lasteqinlhsrhsx}
\end{align}
(\ref{lasteqinlhsrhsx}) holds because $x \le 1$ and $f(v_2)$ is an MHR function. Hence, the optimal signal ratio $x^* \le 1$ must hold. It is easy to see the signal ratio is non-negative. Thus, the lemma is proved.


\end{proof}

With Lemma \ref{lemmas2s1lessthan1}, the following properties are useful in the construction of our simple signaling scheme.
\begin{lemma}\label{propertyxlconstruction}
Given a CTR vector $r = (1, l)$,  the optimal signal ratio $x = \frac{s_2}{s_1}$ of  (\ref{expectedrevenuegivenrs}) depends on $l$, denoted as $x(l)$. We have the following properties.
\begin{properties}[0]{P}
\item $x(0) > 0$.
\item The optimal signal ratio $x(l) > l$ for $l \in [0, 1)$ and $x(1) = 1$.
\item The optimal expected revenue (\ref{expectedrevenuegivenrs}) is increasing in $l$ for $l < 1$.
\end{properties}
\end{lemma}
\begin{proof}
We prove the properties as follows.

{\color{blue}P1}: If $l = 0$, then (\ref{expectedrevenuegivenrs}) equals $\int_0^c v_2 x\int_{v_2x}^c f(v_1)dv_1 f(v_2) dv_2 \ge 0$. The revenue (\ref{expectedrevenuegivenrs}) would be $0$ if $x = 0$, which is not optimal. Thus, $x(0) > 0$.

{\color{blue}P2}: We know that the expected revenue (\ref{expectedrevenuegivenrs}) can be viewed as a function of $l$ and $x$,
\begin{align*}
g(x, l) 
&= \int_0^c  v_2 x\int_{v_2x}^c f(v_1)dv_1 f(v_2) dv_2 +l\cdot \int_0^{cx}  \frac{v_1}{x} \int_{\frac{v_1}{x}}^c f(v_2)dv_2 f(v_1) dv_1\\
&= \int_0^c v_2 x\int_{v_2x}^c f(v_1)dv_1 f(v_2) dv_2 + l \cdot \int_0^{c}  \int_{0}^{v_2 x} \frac{v_1}{x} f(v_1)dv_1 f(v_2) dv_2
\end{align*}
We obtain the partial derivative w.r.t. $x$ as
\begin{align}
\frac{\partial g(x, l)}{\partial x}
&=\frac{\partial}{\partial x} \Big(  \int_0^c v_2 x\int_{v_2x}^c f(v_1)dv_1 f(v_2) dv_2  + l \cdot \int_0^{c}  \int_{0}^{v_2 x} \frac{v_1}{x} f(v_1)dv_1 f(v_2) dv_2 \Big) \notag\\
&= \int_{0}^c v_2 f(v_2) \Big[ \int_{v_2 x}^c f(v_1) dv_1 - xf(v_2x)v_2\Big] dv_2 \notag\\
&\quad + l\cdot \Big( \int_{0}^c f(v_2)\cdot \Big( -\frac{1}{x^2} \int_0^{v_2x} v_1f(v_1) dv_1 + \frac{v_2x}{x} f(v_2x)v_2 \Big)  dv_2 \Big) \notag\\
&= \int_{0}^{c} v_2 f(v_2) \int_{v_2x}^c f(v_1)dv_1 dv_2 -\int_{0}^c v_2^2 f(v_2)f(v_2x)xdv_2 \notag\\
&\quad -\frac{l}{x^2} \int_{0}^c f(v_2) \int_0^{v_2 x} v_1 f(v_1) dv_1 dv_2 +l \int_{0}^c f(v_2) v_2 f(v_2 x)v_2 dv_2 \label{dgx_dx}
\end{align}
Let $h(x, l)$ denote the partial derivative $\frac{\partial g(x, l)}{\partial x}$. We want to show that given $0 \le l <1$, the derivative $h(x, l) > 0$ holds in the range $x \in [0, l]$.
When $x=0$, we have $h(0, l) > 0$ by simple calculation. Thus, we assume $x> 0$ and firstly have
\[l \int_{0}^c f(v_2) v_2 f(v_2 x)v_2 dv_2 \ge \int_{0}^c v_2^2 f(v_2)f(v_2x)xdv_2 ,\]
which holds because $l \ge x$; and secondly we have
\begin{small}
\begin{align}
 &\int_{0}^{c} v_2 f(v_2) \int_{v_2x}^c f(v_1)dv_1 dv_2 > \frac{l}{x^2} \int_{0}^c f(v_2) \int_0^{v_2 x} v_1 f(v_1) dv_1 dv_2\notag\\
 \Longleftrightarrow& \int_{0}^{c} v_2 x f(v_2) \int_{v_2x}^c f(v_1)dv_1 dv_2  > l \int_0^{cx} \frac{v_1}{x} f(v_1) \int_{\frac{v_1}{x}}^{c} f(v_2) dv_2  dv_1 \label{lasteqproperty_lasttoend} \\
 \Longleftrightarrow& \int_{0}^{cx} \frac{t}{x} f(\frac{t}{x}) \int_{t}^c f(v_1)dv_1 dt > l \int_0^{cx} \frac{v_1}{x} f(v_1) \int_{\frac{v_1}{x}}^{c} f(v_2) dv_2  dv_1 \label{lasteqproperty}
\end{align}
\end{small}where (\ref{lasteqproperty_lasttoend}) to (\ref{lasteqproperty}) is by letting $t=v_2 x$. The strict inequality (\ref{lasteqproperty}) holds because 1) $l < 1$, and 2) an  analysis similar to  (\ref{eq14lhsrhsgxhy}), which holds with weak inequality. 
From the analysis, we know that for $0 \le l <1$, $h(x, l) > 0$  holds for $x \in [0, l]$ which implies that given $l<1$, $g(x, l)$ is monotone increasing in $x \in [0, l]$. Also, we have $g(x, l)$ is continuous by the fact that $f(v)$ is continuous. Both together imply that we must have the optimal signal ratio  $x(l) > l$ for $l \in [0, 1)$.

If $l =1$, due to the fact that the optimal signal ratio $x(l) \le 1$ (Lemma \ref{lemmas2s1lessthan1}) and $h(x, 1) > 0$ for any $x <1$, we  have $h( 1,  1) = 0$, and thus $x(1) = 1$.

{\color{blue}P3}: Given a click-through-rate $l$, we denote the maximum expected revenue achieved as $g(l) = \max_x g(x, l)$. Given any $l_1 < l_2$, by substituting $x(l_1)$ into $g(x, l_2)$, we have  $g(l_2) \ge g(x(l_1), l_2)$. Because the second part $  \int_0^{cx} \frac{v_1}{x} \int_{v_1/x}^c f(v_2)dv_2 f(v_1) dv_1\ge 0$  in (\ref{expectedrevenuegivenrs}) and $l_2 > l_1$, we must also have $g(x(l_1), l_2) \ge g(x(l_1), l_1) = \max_x g(x, l_1)$. Hence, the 
optimal expected revenue (\ref{expectedrevenuegivenrs}) is increasing in $l$.

\end{proof}

\section{Discussions on Implications of Lemma \ref{lemmatemperary}}\label{appen_discuss_implic_osigratio}
\jjr{
The seller gains more revenue by creating more intensive competition between bidders by revealing partial information. Recall that in our mechanism, the buyers are ranked by the products of values and CTRs. When applying full-information revelation, it may happen that a low-value (on average) bidder with a higher CTR always wins the auction over a high-value (on average) bidder with a lower CTR. This prohibits the seller from gaining revenue from the high-value bidders. Such undesired cases are because the environment lacks competition. By manipulating the auction through information design, the seller can collect revenue from the high-value bidders. 

We illustrate this observation for our problem using the symmetric environment with two bidders. Suppose the CTR is $r=(h=1, l)$. The signal ratio under full-information revelation is exactly $l$. However, Lemma \ref{lemmatemperary} shows that the optimal signal ratio $s=\frac{s_2}{s_1}>l$. This actually implies a more intensive competition between buyers. By E.q. (\ref{expectedrevenuegivenrs}), buyer $1$ needs to pay a higher price (since $\frac{v_2s_2}{s_1}$ increases in signal ratio $s=\frac{s_2}{s_1}$) but gets a lower chance ($\int^c_{\frac{v_2s_2}{s_1}} f(v_1)dv_1$ decreases) for winning. In comparison, buyer $2$ increases the chance of winning but pays a lower price. Hence, we can see that the competition is more intensive than a full-information environment. The optimal signal ratio achieves some trade-off between the two buyers’ payment changes so that the seller's total revenue is increased by partially revealing information.
}

\section{Missing Proofs in Section \ref{section4symmetricenvironments}} \label{missingproof_sectionsconstructingsignalingscheme}

\subsection{Verifying Signaling Scheme and Choosing $p_0$} \label{app_verifyingsignalingscheme}
\begin{lemma}\label{verifyingsignalingscheme}
By carefully choosing $p_0$, the constructed signaling scheme is valid and calibrated.
\end{lemma}
\begin{proof}
By the construction, we can see that the signals used are valid, i.e., $\sigma_i \in [l, 1], \forall i$. Let $\Sigma = \{\sigma_0, \sigma_1, \cdots, \sigma_{K-1}, \sigma_K\}$. The next part of the proof is to verify that the calibration constraint holds for any signal $s_i$. Without loss of generality, we consider the bidder $1$. The signals sent to the  bidder $1$ are calibrated if  (\ref{calibrationconstraintdiscrete}) holds. When sending signal $\sigma_i, i \in \{1, 2, \cdots, K-1\}$ to bidder $1$, we have (\ref{calibrationconstraintdiscrete}) as
\begin{align*}
&\sum_{\sigma \in \Sigma} \lambda\big((l, 1)\big)\pi\big((\sigma_i, \sigma)|(l, 1)\big)(l - \sigma_i) +  \lambda\big((1, l)\big)\pi\big((\sigma_i, \sigma)|(1, l)\big)(1 - \sigma_i) \\
&= p_i(l-\sigma_i) + p_{i-1}(1-\sigma_i)\\
&= p_{i-1}\frac{1-\sigma_i}{\sigma_i -l} \cdot (l-\sigma_i) + p_{i-1}(1-\sigma_i) = 0
\end{align*}
Similarly, we can verify the calibration constraint for signals sent to bidder $2$. By the construction, we know that when bidder $1$ or bidder $2$ observes signal $\sigma_K = 1$, the click-through-rate can only be $h = 1$, so the calibration constraint also holds. 

Next, we verify the calibration constraint for the signal $\sigma_0$. Without loss of generality, we only need to consider bidder $1$. There are only three signals containing $\sigma_0$ sent to bidder $1$, which are $\pi((\sigma_0, \sigma_0) | (l, 1))$, $\pi((\sigma_0, \sigma_1) | (l, 1))$ and $\pi((\sigma_0, \sigma_0) | (1, l))$. Hence, the calibration constraint (\ref{calibrationconstraintdiscrete}) holds since
\begin{align*}
&\sum_{\sigma \in \Sigma} \lambda\big((l, 1)\big)\pi\big((\sigma_0, \sigma)|(l, 1)\big)(l - \sigma_0) +  \lambda\big((1, l)\big)\pi\big((\sigma_0, \sigma)|(1, l)\big)(1 - \sigma_0)\\
&= p_0 (l-\sigma_0) + z(l-\sigma_0) + z(1-\sigma_0)\\
&= p_0 (l-\sigma_0)  + p_0 \frac{\sigma_0 -l}{l+1-2\sigma_0} (l+1 -2\sigma_0)\\
&= 0
\end{align*}


Next, we need to show the probability mass $z \ge 0$. If $\sigma_0 = l$, obviously $z = 0$ by (\ref{relationzx0}). If $\sigma_0 \neq l$, the following lemma is sufficient to show that $z >0$, which indicates that the signaling scheme is valid.

\begin{lemma}\label{lemmalp1m2s}
If $\sigma_0 \neq l$, then $\frac{\sigma_0 - l}{l+1 - 2\sigma_0} >0$.
\end{lemma}
\begin{proof}
By the construction and $0<x(l) < 1$, we have $\sigma_0 = x(l)^K > l$ for some $K \ge K_0 \ge 1$ and $\sigma_0 \cdot x(l) = x(l)^{K+1} \le l$. Thus, we only need to prove $l+1 -2\sigma_0 > 0$, which is equivalent to proving 
\begin{equation*}
  1 - \sigma_0 > \sigma_0 - l \quad
    \Leftrightarrow  \quad 1 - x(l)^K > x(l)^K - l
\end{equation*}
Then it is sufficient to show $1 - x(l)^K > x(l)^K - x(l)^{K+1} \ge x(l)^K - l$. The first inequality $1 - x(l)^K > x(l)\cdot (x(l)^{K-1} - x(l)^{K})$ holds because $x(l) < 1$ and $x(l)^{K-1} \le 1$. The second inequality holds because $x(l)^{K+1}\le l$. 
\end{proof}

Finally, we need to choose $p_0$ so that $\sum_s p((l, h), s) = \sum_{s} p((h, l), s) = \frac{1}{2}$, i.e., the probability constraint is satisfied. Without loss of generality, by considering the state $(l, 1)$, we have 
\begin{equation}\label{12equationsolution}
    z + p_0 + p_1 + \cdots + p_{K-1} = \frac{1}{2}
\end{equation}
Thus, we choose $p_0 = \frac{1}{2}\frac{1}{(z/p_0 + 1 + \sum_{i=1}^{K-1}p_i/p_0)}$.
Note that by (\ref{xkprobabilitymass}) and (\ref{relationzx0}), we know $\frac{z}{p_0}$ and $\frac{p_i}{p_0}$ are some constants given $l$ and the optimal signal ratio $x(l)$.

\end{proof}

\subsection{Proof of Observation \ref{intersectionpointobservation}}

\begin{proofof}{Observation \ref{intersectionpointobservation}}
In Lemma \ref{propertyxlconstruction}, we know $x(1) = 1$. Thus, the optimal signal ratio function $x(l)$ and line $x=l$ must intersect at $(1, 1)$. Since $x = x(l)^k$ is convex, there are at most two points intersecting with line $x=l$ (by the definition of convexity). Hence, the other point intersecting is $(l[k] , l[k])$ where $l[k]$ is one solution to equation $x(l)^k = l$ with $l[k] \neq 1$. Because $x(l) < 1$ for $l <1$, $x(l)^k$ decreases as $k$ increases. The curve $x = x(l)^k$ is  below the curve $x = x(l)^{k-1}$. Hence, $l[k]$ decreases.
\end{proofof}

\subsection{Proof of Lemma \ref{convexmontonefunction}}
\label{derivationoflxwithvirtualval}

\begin{proofof}{Lemma \ref{convexmontonefunction}}
We first notice from (\ref{dgx_dx}) that the partial derivative $\frac{\partial g(x, l)}{\partial x}$ is a continuous function in $x$. It implies that for $l \in [0, 1)$ with $x(l) < 1$, it must be $\frac{\partial g(x, l)}{\partial x}|_{x = x(l)} = 0$. We first simplify the derivative $\frac{\partial g(x, l)}{\partial x}$ in (\ref{dgx_dx}).
\begin{align*}
&\int_{0}^{c} v_2 f(v_2) \int_{v_2x}^c f(v_1)dv_1 dv_2 -\int_{0}^c v_2^2 f(v_2)f(v_2x)xdv_2 \\
&=\int_{0}^{c} v_2 f(v_2)( (1-F(v_2x)) - v_2 x f(v_2 x)) dv_2 \\
&=\int_{0}^{c} v_2 f(v_2)f(v_2 x)( \frac{ 1-F(v_2x)}{f(v_2x)} - v_2 x ) dv_2
\end{align*}
Also, we have (we use $t = v_2 x$ and $s = \frac{t}{x}$)
\begin{align*}
& \int_{0}^c f(v_2) v_2 f(v_2 x)v_2 dv_2 -\frac{1}{x^2} \int_{0}^c f(v_2) \int_0^{v_2 x} v_1f(v_1)  dv_1 dv_2\\
&=  \int_{0}^{cx}  f(t) f(\frac{t}{x}) (\frac{t}{x})^2 \frac{1}{x} dt - \frac{1}{x}\int_{0}^{cx} \frac{v_1}{x} f(v_1) \int_{\frac{v_1}{x}}^c  f(v_2) dv_2 dv_1 \\
&= \int_{0}^{cx}  f(t) f(\frac{t}{x}) (\frac{t}{x})^2 \frac{1}{x} dt - \frac{1}{x}\int_{0}^{cx} (1-F(\frac{t}{x}))  f(t) \frac{t}{x} dt \\
&= \int_{0}^{cx} f(t) \frac{t}{x^2}  (f(\frac{t}{x})\frac{t}{x} - (1-F(\frac{t}{x}))) dt \\
&= \int_{0}^{cx} f(t) \frac{t}{x^2} f(\frac{t}{x}) (\frac{t}{x} - \frac{(1-F(\frac{t}{x}))}{f(\frac{t}{x})}) dt \\
&=  \int_{0}^{c}  f(sx) s f(s) (s - \frac{(1-F(s))}{f(s)}) ds 
\end{align*}
Combining these two equations, we have (replace $s$ and $v_2$ with $v$)
\begin{equation}\label{simplifieddgdx}
\begin{split}
\frac{\partial g(x, l)}{\partial x}=& \int_0^{c} v f(v)f(v x)( \frac{ (1-F(vx))}{f(vx)} -l\frac{ (1-F(v))}{f(v)} +  lv -  v x ) dv\\
=& \int_0^{c} v f(v)f(v x)( l\phi(v) - \phi(vx) ) dv
\end{split}
\end{equation}
where $\phi(x) = x - \frac{1-F(x)}{f(x)}$ is a regular function because $f(x)$ is an MHR function. Since $x(l)$ is continuous and convex, if $x(l)$ is not a strictly increasing function, then there must exist some $x \in [0, 1)$ such that for some $l \neq l'$,  $\frac{\partial g(x, l)}{\partial x} = 0$ and  $\frac{\partial g(x, l')}{\partial x} = 0$ hold simultaneously. Hence, by subtracting $\frac{\partial g(x, l')}{\partial x}$ from $\frac{\partial g(x, l)}{\partial x}$, we have
\begin{equation*}
\int_0^{c} v f(v)f(v x)(l - l')\phi(v) dv = 0
\end{equation*}
which implies that (since $l - l' \neq 0$)
\[\int_0^{c}v f(v)f(v x)\phi(v) dv = 0\]
This further implies that
\[\int_0^{c} v f(v)f(v x) \phi(vx) dv = 0\]
Hence, there must exist some $x \in [0, 1)$, $\frac{\partial g(x, l)}{\partial x}= 0$ for any $l$. However, this contradicts the result that given any $x \in [0, 1)$ and any $l \in [x, 1]$, we should have $\frac{\partial g(x, l)}{\partial x} > 0$ (shown in the proof of   {\color{blue}P2} of Lemma \ref{propertyxlconstruction}). Thus, $x(l)$ must be an increasing function if $x(l)$ is convex.

\end{proofof}

\subsection{Proof of Lemma \ref{opensetmonotonictiy}}\label{app_opensetmonotonictiy}

\begin{proofof}{Lemma \ref{opensetmonotonictiy}}
Since $z\cdot S(k, l) = \frac{1}{2}$, we prove the lemma by showing that $\frac{l+1-2\sigma_0}{\sigma_0 -l}$ and $\frac{1-\sigma_i}{\sigma_i -l}$ are increasing in $l$. Since $\sigma_0 = x(l)^k$, we take the derivative with respect to $l$ and show that it is greater than $0$.
\begin{align*}
\frac{d}{dl}\frac{l+1-2x(l)^k}{x(l)^k -l} 
&=\frac{(l+1-2x(l)^k)'(x(l)^k -l) - (x(l)^k -l)'(l+1-2x(l)^k)}{(x(l)^k -l)^2}\\
&=\frac{(1-2kx(l)^{k-1}x'(l))(x(l)^k -l) - (kx(l)^{k-1}x'(l) -1)(l+1-2x(l)^k)}{(x(l)^k -l)^2}
\end{align*}
Therefore, we need to show $(1-2kx(l)^{k-1}x'(l))(x(l)^k -l) - (kx(l)^{k-1}x'(l) -1)(l+1-2x(l)^k) \ge 0$. We further know that
\begin{align*}
 h(l)
 &= (1-2kx(l)^{k-1}x'(l))(x(l)^k -l) - (kx(l)^{k-1}x'(l) -1)(l+1-2x(l)^k)\\
 &= x(l)^k -l - 2kx(l)^{k-1}x'(l)x(l)^k + 2kx(l)^{k-1}x'(l)l \\
 &\quad - kx(l)^{k-1}x'(l)l - kx(l)^{k-1}x'(l) + 2kx(l)^{k-1}x'(l)x(l)^k + l+1-2x(l)^k \\
 &= -x(l)^k + 1 + kx(l)^{k-1}x'(l)l - kx(l)^{k-1}x'(l) \\
 &=x(l)^{k-1}\Big(lkx'(l) - x(l)- kx'(l)\Big) + 1
\end{align*}
We can observe that $h(l=1) = 0$. If we can show that $h(l)$ is decreasing in $l$, then we have $h(l) \ge 0, \forall l \in (l[k+1], l[k]]$. This is done by 
\begin{align*}
 \frac{dh(l)}{dl} &= (k-1)x'(l)x(l)^{k-2}\Big(lkx'(l) - x(l)- kx'(l)\Big)  + x(l)^{k-1}\Big(kx'(l) + lkx''(l) - x'(l)- kx''(l)\Big)\\
 &= (k-1)x'(l)x(l)^{k-2}lkx'(l) - (k-1)x'(l)x(l)^{k-2}kx'(l)  + x(l)^{k-1}\Big(lkx''(l) - kx''(l)\Big) \\
 &= x(l)^{k-2}\Big(k(k-1)(l-1)x'(l)^2 + x(l)(l-1)kx''(l)\Big)\\
 &\le 0
\end{align*}
where the last inequality is due to $l < 1$, $x(l)$ is convex and $k \ge 1$. Hence, $h(l)$ is decreasing in $l$. Thus, $\frac{l+1-2\sigma_0}{\sigma_0 -l}$ is increasing in $l$.

We similarly prove $\frac{1-\sigma_i}{\sigma_i -l}$ is increasing in $l$. Since $\sigma_i = x(l)^{k-i}$, by letting $p=k-i \ge 1$, we take the derivative with respect to $l$ as
\begin{align*}
 \frac{d}{dl} \frac{1-x(l)^p}{x(l)^p - l}
 &=\frac{(1-x(l)^p)'(x(l)^p - l) -(x(l)^p - l)'(1-x(l)^p) }{(x(l)^p - l)^2} \\
 &=\frac{-px(l)^{p-1}x'(l)(x(l)^p - l) -(px(l)^{p-1}x'(l) - 1)(1-x(l)^p) }{(x(l)^p - l)^2} \\
  &= \frac{\splitfrac{-px(l)^{p-1}x'(l)x(l)^p + px(l)^{p-1}x'(l)l}{ -(px(l)^{p-1}x'(l)- px(l)^{p-1}x'(l) x(l)^p - 1+x(l)^p)} }{(x(l)^p - l)^2}\\
  &= \frac{px(l)^{p-1}x'(l)l -px(l)^{p-1}x'(l)+ 1-x(l)^p) }{(x(l)^p - l)^2}
\end{align*}
Thus, we only need to show that
$px(l)^{p-1}x'(l)l -px(l)^{p-1}x'(l)+ 1-x(l)^p)  =  x(l)^{p-1}\Big(plx'(l) - x(l) -px'(l) \Big) +1$  is greater than $0$. Notice that this is analogous to $h(l)$ defined above. Thus,  $\frac{1-\sigma_i}{\sigma_i -l}$ is increasing in $l$. Therefore, we know that $z$ is decreasing in $l \in (l[k+1], l[k]].$

Finally, by Lemma \ref{lemmalp1m2s},  we  can show $z>0$ when $l\neq \sigma_0$.
\end{proofof}


\subsection{Proof of Proposition \ref{propositionexponentialdistribution}}\label{proof_of_propositionexponentialdistribution}

\begin{proofof}{Proposition \ref{propositionexponentialdistribution}}
Since the CTRs of bidders are exchangeable, we can separately consider the signaling scheme for each pair of CTR vectors. It is without loss of generality to consider a pair of CTR vectors $(l, 1)$ and $(1, l)$ such that $\lambda(r = (l, 1)) = \lambda(r = (1, l)) = \frac{1}{2}$.  The expected revenue by (\ref{expectedrevenuegivenrs}) is realized as (for the click-through rate $(1, l)$)
\begin{align*}
R(r, s) &= \int_0^\infty \frac{v_2 s_2}{s_1} \cdot \int_{\frac{v_2s_2}{s_1}}^\infty \lambda \cdot e^{-\lambda v_1} dv_1  \cdot \lambda \cdot e^{-\lambda v_2} dv_2 \\
&\quad \quad \quad + l \cdot \int_0^\infty  \frac{v_1 s_1}{s_2}\cdot \int_{\frac{v_1s_1}{s_2}}^\infty \lambda \cdot e^{-\lambda v_2} dv_2  \cdot \lambda \cdot e^{-\lambda v_1} dv_1
\end{align*} 
We further simplify the formula for the first term (simplification for the second term is similar)
\begin{align*}
 \int_0^\infty \frac{v_2 s_2}{s_1}\cdot \int_{\frac{v_2s_2}{s_1}}^\infty \lambda \cdot e^{-\lambda v_1} dv_1  \cdot \lambda \cdot e^{-\lambda v_2} dv_2 
&= \int_0^\infty \frac{v_2 s_2}{s_1}\cdot  e^{-\lambda \frac{v_2s_2}{s_1}}   \cdot \lambda \cdot e^{-\lambda v_2} dv_2\\
&=  \frac{s_2}{s_1} \int_0^\infty (\lambda v_2 )\cdot  e^{-\lambda v_2 (\frac{s_2}{s_1} + 1)}  dv_2 \\
&= \frac{s_2}{s_1} \cdot \frac{1}{(\frac{s_2}{s_1} + 1)} \int_0^\infty \lambda v_2 (\frac{s_2}{s_1} + 1) \cdot  e^{-\lambda v_2 (\frac{s_2}{s_1} + 1)}  dv_2\\
&= \frac{ s_2}{s_1} \frac{1}{(\frac{s_2}{s_1} + 1)} \cdot \frac{1}{\lambda(\frac{s_2}{s_1} + 1)}  
\end{align*} 
Finally, we have
\begin{align*}
R(r, s) = \frac{1+l}{\lambda} \cdot \frac{s_1s_2}{(s_1+s_2)^2}
\end{align*} 
which is maximized at the optimal signal ratio $\frac{s_1}{s_2} = 1$. That is, $x(l) = 1$ for all $l \in [0, 1]$. By the calibration constraint, {we know that the seller should send $(\frac{1+l}{2}, \frac{1+l}{2})$ with probability $1$ for any pair of CTR vectors}, which implies revealing no information is optimal.
\end{proofof}

\subsection{Proof of Lemma \ref{proposition46_lemma_2}}\label{proof_of_theoremunformdistribution_lemma}
\begin{proofof}{Lemma \ref{proposition46_lemma_2}}
By Lemma \ref{lemmas2s1lessthan1}, the optimal signal ratio $\frac{s_2}{s_1} \le 1$. Therefore,  the expected revenue  (\ref{expectedrevenuegivenrs}) is realized as (for the click-through rate $(1, l)$)
\begin{align}
R(r, s) &=  \int_0^c \frac{v_2s_2}{s_1} \int_{\frac{v_2s_2}{s_1}}^c \frac{1}{c}dv_1 \frac{1}{c} dv_2 + l\cdot \int_0^{c\frac{s_2}{s_1}} \frac{v_1s_1}{s_2}\int_{\frac{v_1s_1}{s_2}}^c \frac{1}{c} dv_2 \frac{1}{c} dv_1 \nonumber\\
&= c\cdot \Big(\frac{1}{2}\frac{s_2}{s_1} - \frac{1}{3}\frac{s_2^2}{s_1^2} + \frac{1}{6}l\frac{s_2}{s_1} \Big)
\label{uniformdistributionexpected}
\end{align}
By solving (\ref{uniformdistributionexpected}), we know that the optimal signal ratio is $x(l) = \frac{s_2}{s_1} = \frac{3+l}{4}$, which is a linear. In (\ref{uniformdistributionexpected}), the expected revenue is only determined by the signal ratio. Hence, the approximation ratio is obtained by comparing the expected revenue obtained from $\frac{\sigma_0}{\sigma_0} = 1$ and that from the optimal signal ratio $\frac{3+l}{4}$. The maximum expected revenue achieved at $\frac{3+l}{4}$ is  $c\cdot\frac{(3+l)^2}{48}$, while $\frac{\sigma_0}{\sigma_0} = 1$ leads to expected revenue $c\cdot\Big(\frac{1}{6} + \frac{1}{6}l\Big)$.  The ratio between the revenue of sending signal $(\sigma_0, \sigma_0)$ and the maximum possible revenue is at least $\frac{8}{9}$.
\end{proofof}

\subsection{Proof of Lemma \ref{lemmaskdecreasing}}\label{proof_of_lemma11finalsection}

\begin{proofof}{Lemma \ref{lemmaskdecreasing}}
By Definition \ref{definitioninitiaointersection},  we first write down the explicit formula for $S(k, l[k+1])$ and $S(k-1, l[k])$.
\begin{small}
\begin{align*}
   S(k-1, l[k])
   &= 1+  \Big(\frac{1-x(l[k])^{k-1}}{x(l[k])^{k-1} -l[k]} - 1\Big) + \Big(\frac{1-x(l[k])^{k-1}}{x(l[k])^{k-1} -l[k]} - 1\Big)\cdot \frac{1-x(l[k])^{k-2}}{x(l[k])^{k-2} -l[k]} \\
    &\quad +\cdots + \Big(\frac{1-x(l[k])^{k-1}}{x(l[k])^{k-1} -l[k]} - 1\Big) \cdot \prod_{i=1}^{k-2}\frac{1-x(l[k])^{k-1-i}}{x(l[k])^{k-1-i} -l[k]}
\end{align*}
\end{small}
\begin{small}
\begin{align*}
   S(k, l[k+1])
   &= 1+ \Big(\frac{1-x(l[k+1])^{k}}{x(l[k+1])^{k} -l[k+1]} - 1 \Big)  + \Big(\frac{1-x(l[k+1])^{k}}{x(l[k+1])^{k} -l[k+1]} - 1 \Big)\cdot \frac{1-x(l[k+1])^{k-1}}{x(l[k+1])^{k-1} -l[k+1]} \\
    &\quad +\cdots + \Big(\frac{1-x(l[k+1])^{k}}{x(l[k+1])^{k} -l[k+1]} - 1 \Big) \cdot  \prod_{i=1}^{k-2}\frac{1-x(l[k+1])^{k-i}}{x(l[k+1])^{k-i} -l[k+1]} \\
    &\quad \quad\quad \quad+  \Big(\frac{1-x(l[k+1])^{k}}{x(l[k+1])^{k} -l[k+1]} - 1 \Big) \cdot  \prod_{i=1}^{k-1}\frac{1-x(l[k+1])^{k-i}}{x(l[k+1])^{k-i} -l[k+1]}
\end{align*}
\end{small}Hence, to prove the lemma, it is sufficient to show that  for any integer $i \le k-1$, we have 
\begin{equation}\label{inequalitytoprove}
\frac{1-x^{i}(l[k])}{x^{i}(l[k]) - l[k]} <\frac{1-x^{i+1}(l[k+1])}{x^{i+1}(l[k+1]) - l[k+1]}
\end{equation}
where $x^{i}(l[k]) > l[k]$ and $x^{i+1}(l[k+1]) > l[k+1]$. The inequality can be equivalently rewritten as
\begin{small}
\begin{align*}
&\frac{1-x^{i}(l[k])}{x^{i}(l[k]) - l[k]} <\frac{1-x^{i+1}(l[k+1])}{x^{i+1}(l[k+1]) - l[k+1]} \\
\Longleftrightarrow\quad & (1-x^{i}(l[k]))(x^{i+1}(l[k+1]) - l[k+1])  < (1-x^{i+1}(l[k+1])) (x^{i}(l[k]) - l[k]) \\
\Longleftrightarrow\quad & x^{i+1}(l[k+1]) - l[k+1] + x^{i}(l[k]) l[k+1]  < x^{i}(l[k]) - l[k] + x^{i+1}(l[k+1]) l[k] \\
\Longleftrightarrow\quad&1 - x^{i}(l[k]) - l[k+1] +  x^{i}(l[k]) l[k+1]  < 1 - x^{i+1}(l[k+1]) - l[k] +  x^{i+1}(l[k+1]) l[k]\\
\Longleftrightarrow \quad & (1-l[k+1])(1-x^i(l[k])) < (1-l[k])(1- x^{i+1}(l[k+1]))\\
\Longleftrightarrow \quad & \frac{1-l[k+1]}{1-l[k]} <  \frac{1- x^{i+1}(l[k+1])}{1-x^i(l[k])}
\end{align*}
\end{small}Recall that $l[k] = x(l[k])^k$ and $l[k+1] = x(l[k+1])^{k+1}$.  Then, the inequality becomes
\begin{equation}\label{inequalityconverted}
    \frac{1-l[k+1]}{1-l[k]} = \frac{1-x^{k+1}(l[k+1])}{1 - x^k(l[k])} < \frac{1 - x^{i+1}(l[k+1])}{1 - x^i(l[k])}
\end{equation}
Let $y_k = x(l[k])$ and note that $y_{k+1} < y_k < 1$ by the monotonicity in Lemma \ref{convexmontonefunction} and $l[k+1] < l[k]$ in Observation \ref{intersectionpointobservation}. Define 
\[f(x) = \frac{1 - y_{k+1}^{x+1}}{1 - y_{k}^{x}}.\]
Take the derivative with respect to $x$ as 
\[\frac{d f(x)}{d x} = \frac{y_{k+1}^{x+1} \ln (y_{k+1}) \big( y_k^{x} - 1\big) - \big( y_{k+1}^{x+1} - 1\big)  y_k^{x}  \ln (y_k) }{\big(1 - y_k^{x} \big)^2}\]
Recall that in (\ref{inequalityconverted}), $i<k$. If  $\frac{d f(x)}{d x} < 0$, then  the inequality (\ref{inequalityconverted}) is satisfied. In other words, we want $\frac{d f(x)}{d x} <0$ to hold for any $x$. Since $y_{k+1} < y_k$, we let $a$ represent $y_{k+1}$ and $b$ represent $y_k$, and prove the following claim.

\textbf{Claim 1:} {\it Given any $1 > b> a>0$,  the following inequality holds for any $x > 0$.
\[a^{x+1} (\ln a) (b^x - 1) - (a^{x+1} - 1) b^x \ln b < 0\]}

Note that \textbf{Claim 1} has an equivalent alternative form (by replacing $x$ with $k$ and $b$ with $x$),

\textbf{Claim 2:} {\it Given any $1 >  a>0$ and any $k >0$,  the following inequality holds for any $1 >  x  > a$.
\[a^{k+1} (\ln a) (x^k - 1) - (a^{k+1} - 1) x^k \ln x < 0\]}

If \textbf{Claim 2} is true,  \textbf{Claim 1} is immediately   true.

The inequality in \textbf{Claim 2} can be further simplified to
\[\frac{a^{k+1} \ln a}{a^{k+1} - 1} \le \frac{x^k \ln x}{x^k - 1}\]
We further note that (since $a < 1$, we know $\frac{\ln a}{a^{k+1} - 1} > 0$ and $a^k > a^{k+1}$)
\[\frac{a^{k+1} \ln a}{a^{k+1} - 1} < \frac{a^{k} \ln a}{a^{k+1} - 1} < \frac{a^{k} \ln a}{a^{k} - 1}\]
Then, it is sufficient to show that 
\begin{align}
  &\frac{a^{k+1} \ln a}{a^{k+1} - 1} < \frac{a^{k} \ln a}{a^{k} - 1} = \frac{1}{k} \frac{a^{k} \ln a^k}{a^{k} - 1} = \frac{1}{k} \frac{a^{k} \cdot k \cdot \ln a}{a^{k} - 1}  \le \frac{x^k \ln x}{x^k - 1} =  \frac{1}{k} \frac{x^k \cdot k \cdot \ln x}{x^k - 1} = \frac{1}{k} \frac{x^k  \ln x^k}{x^k - 1}
\end{align}
Let $g(y) = \frac{1}{k} \frac{y\ln y}{y - 1}$. Since $1> x > a$ and $x^k > a^k$ for $k >0$, we only need to show $g(y)$ is increasing in $y \in (0, 1)$, which is true because the derivative $\frac{d g(y)}{dy}$ is larger than $0$, i.e.,
\[\frac{d g(y)}{dy} = \frac{y - \ln y - 1}{(y - 1)^2} > 0\]
Thus, \textbf{Claim 2} is true, and so is \textbf{Claim 1}. Therefore, the derivative $\frac{d f(x)}{d x} < 0$ which immediately implies (\ref{inequalityconverted}) holds. Finally, the inequality (\ref{inequalitytoprove}) holds for any $i$. Hence, we have $S(k-1, l[k]) \le S(k, l[k+1])$. The lemma is then proved.
\end{proofof}

\section{Example for Verifying the Convexity of Optimal Signal Ratio Function}\label{example_examplecomputexlandinitialnumber}
We provide one more example to explain how to computationally verify the convexity of $x(l)$ and find the initial number $K_0$. 

    Assume that the density function is $f(v) = 12v(1-v)$, supported in the bounded interval $[0, \frac{1}{2}]$. It is not hard to verify that the distribution is MHR. The expected revenue (\ref{expectedrevenuegivenrs}) is rewritten as 
\begin{align*}
R(l, x)
&=  \int_0^{\frac{1}{2}} v_2 x \int_{v_2 x}^\frac{1}{2} f(v_1) dv_1 f(v_2) dv_2 + \int_0^\frac{x}{2} (l \cdot \frac{v_1}{x} )\int_{\frac{v_1}{x}}^\frac{1}{2} f(v_2) dv_2 f(v_1) dv_1 
\end{align*}
where $x = \frac{s_2}{s_1}$. By some calculation, we have $R(l, x) = l(-\frac{3 x^3}{56} + \frac{7 x^2}{40}) + \frac{x^4}{14} - \frac{21 x^3}{80} + \frac{5x}{16}$. However, it would be hard to have an explicit formula for $x(l), l\in [0, 1]$ by optimizing $R(l, x)$. Instead, we found that it is  relatively easy to find its inverse function $l(x)$. Observe that $\frac{\partial R(l, x)}{\partial x}$ is continuous in $x$. Hence, at the optimal signal ratio $x^*$, one must have $\frac{\partial R(l, x)}{\partial x}|_{x=x^*} = 0$. We then obtain an explicit formula for $l(x)$, (a detailed derivation is in Appendix \ref{derivationoflxwithvirtualval})
\[ l(x) = \frac{\int_0^{\frac{1}{2}} v f(v)f(v x)\phi(vx) dv }{\int_0^\frac{1}{2} v f(v)f(v x)\phi(v) dv}\]
where $\varphi(x) = x - \frac{1-F(x)}{f(x)}$ is the virtual valuation. We plot $l(x)$ in Figure \ref{figobservationaaaproofinverse} (A). 

Notice that $l(x)$ is clearly continuous in $x$. Since $x(l)$ is strictly increasing for convex $x(l)$ by Lemma \ref{convexmontonefunction}, $l(x)$ (for $l \in [0, 1]$ and $x \in [x_0, 1]$, where $x_0 \in [0, 1]$ and $l(x_0) = 0$) must be a concave function if $x(l)$ is convex. We can verify that $\frac{\partial^2 l(x)}{\partial x^2} \le 0$ for $x\in [0, 1]$, which implies that $x(l)$ is a convex function (also see Figure \ref{figobservationaaaproofinverse} (A)).

To compute $K_0$, we need to determine whether $x(l)^k$ for some integer $k$ intersects with line $x=l$. In the context of $l(x)$, if $x(l)$ intersects with line $x=l$, then  it is equivalent to the inverse function $l(x^{\frac{1}{k}})$ intersecting with line $x=l$. Hence, we can find $K_0$ by enumerating $k$. The intersection point $l[K_0+1]$ then can be found by binary search, with which we can compute $x(l[K_0+1])$. See Figure \ref{figobservationaaaproofinverse} (B) for an example.

In this example, we have $K_0 = 3$, $l[K_0 + 1]\approx 0.7792$ and $x(l[K_0 + 1]) \approx 0.9395$. By Proposition \ref{theroemboudningapproximation},  we obtain $z^* \approx 0.14$. Hence, the constructed signaling scheme can achieve approximately at least $0.86$ approximation.
\begin{figure}[t]
 \centering 
\scalebox{0.26}{\includegraphics{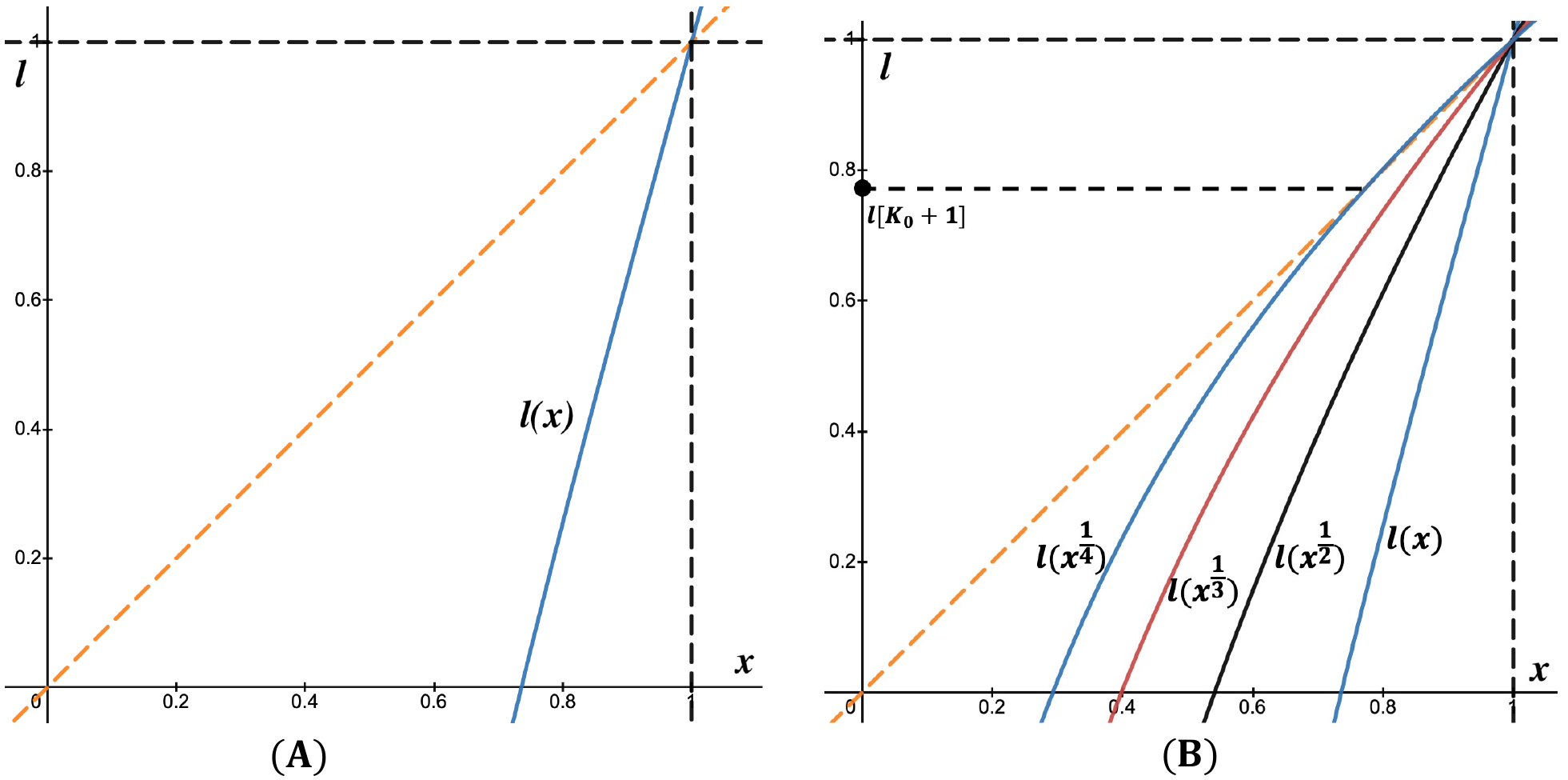}}\caption{$l(x)$ is the inverse function of  $x(l)$. $l(x)$ is concave. $K_0$ is initial number. \label{figobservationaaaproofinverse}}
\end{figure}

\section{Additional Discussions}\label{discussionconvexitymorethan2buyer}

\begin{figure}[h]
 \centering 
\scalebox{0.35}{\includegraphics{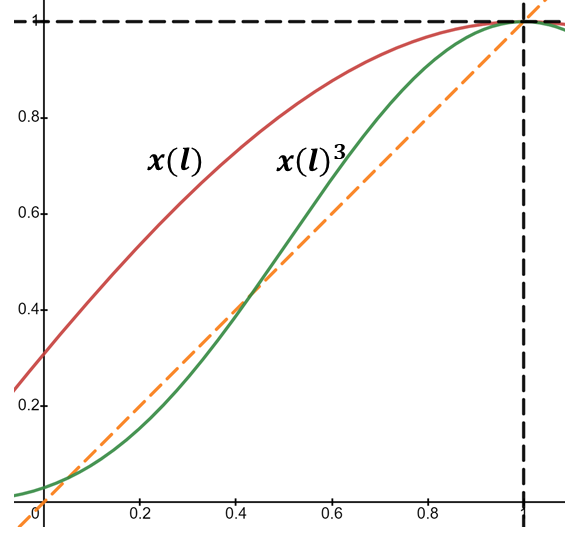}}\caption{Let $x(l)= sin \Big(\frac{4\pi}{10}x + \frac{\pi}{10} \Big)$ be a concave function. As shown, $x(l)^3$ intersects with line $x=l$ more than twice. Also, when an intersection occurs, $x(l)^3$ is not a convex function. \label{concaveexamples}}
\end{figure}

\noindent\textbf{Convexity assumption of $x(l)$.}  
For concave or arbitrary function $x(l)$, Proposition \ref{theroemboudningapproximation} cannot be directly generalized. One requirement of Proposition \ref{theroemboudningapproximation} is that given a convex $x(l)$, $x(l)^k$ will only intersect with line $x=l$ at most twice, and the intersection point $l[k]$ decreases in $k$. However, there exist some concave function $x(l)$ such that $x(l)^k$ may intersect with line $x=l$ more than twice ({see Figure \ref{concaveexamples}} for an example), which makes our analysis not applicable. The problem of constructing an approximate signaling scheme given a general $x(l)$ remains open.

\noindent\textbf{More Than Two Bidders in Symmetric Environments.}  
This problem is of interest and remains open. It shows a very different structure from the case of two bidders. For example, suppose $f(v)$ is an exponential density function over $[0, \infty)$. When there are only two bidders, the optimal signal ratio is $x = 1$, leading to an optimal signaling scheme revealing no information. But it is no longer the case for $3$ bidders. The expected revenue for $3$ bidders is computed as
\begin{align*}
R(r, s)
&= \frac{r_1+r_2}{\lambda}xy^2\Big(\frac{1}{(xy+y)^2} -\frac{1}{(xy+1+y)^2}\Big)\\
&\quad +\frac{r_2+r_3}{\lambda}y\Big(\frac{1}{(y+1)^2} -\frac{1}{(xy+1+y)^2}\Big)\\
&\quad +\frac{r_1+r_3}{\lambda}xy\Big(\frac{1}{(xy+1)^2} -\frac{1}{(xy+1+y)^2}\Big)
\end{align*}
where $x = \frac{s_2}{s_1}$ and $y = \frac{s_3}{s_2}$. By simple calculation, we find that $x = y = 1$ is not the optimal signal ratio. We find it is challenging to generalize the idea of the signaling scheme construction in Section \ref{sectionsconstructingsignalingscheme} to this case. A more complicated signaling scheme is expected.

\end{document}